\documentclass[11pt,leqno]{article}%
\usepackage{amsfonts}
\usepackage{amssymb}
\usepackage{amsmath}
\usepackage{graphicx}
\usepackage{natbib}
\usepackage{epsfig}
\usepackage{subfigure}
\usepackage[left=2.6cm,top=2.4cm,right=2.6cm]{geometry}
\usepackage[doublespacing]{setspace}
\usepackage{fancyhdr}
\usepackage{verbatim}%
\newtheorem{theorem}{Theorem}
\newtheorem{corollary}{Corollary}
\newtheorem{definition}{Definition}

\newtheorem{proposition}{Proposition}
\newtheorem{remark}{Remark}
\newtheorem{example}{Example}
\newtheorem{identity}{Identity}
\newenvironment{proof}[1][Proof]{\noindent\textbf{#1.} }{$\hfill\square$\bigskip}

\newtheorem{defn0}{Definition}

\newtheorem{theorem0}{Theorem}
\newtheorem{proposition0}{Proposition}
\newtheorem{example0}{Example}
\newtheorem{remark0}{Remark}

\newtheorem{identity0}{Identity}
\renewenvironment{definition}{ \flushleft \begin{defn0}}{\end{defn0}}
\renewenvironment{theorem}{\flushleft \begin{theorem0} }{\end{theorem0}}
\renewenvironment{proposition}{ \flushleft \begin{proposition0}}{\end{proposition0}}

\renewenvironment{remark}{\medskip \flushleft \begin{remark0}}{\end{remark0}}

\renewcommand{\cite}{\citet*}
\makeatletter
\renewcommand\section{\@startsection
{section}{1}{0pt}{3pt} {2pt} {\large\bf}}
\renewcommand\subsection{\@startsection
{subsection}{1}{0pt}{1pt} {1pt} {\large\bf}}
\renewcommand\subsubsection{\@startsection
{subsubsection}{1}{0pt}{-2.5ex plus-1ex minus-.2ex} {2.3ex plus.2ex}
{\normalfont\bf}}
\renewcommand{\thesection}{\arabic{section}}
\renewcommand{\thesubsection}{\arabic{section}.\arabic{subsection}}

\numberwithin{equation}{section}

\begin{document}
\bigskip%
\begin{center}
\large{Hedging strategies and minimal variance portfolios}
\large{\\ for European and exotic options in a  L\'{e}vy market}
\vspace{0.2 in}
\large{\\ Wing Yan Yip}
\textit{\\ Department of Mathematics, Imperial College London}
\large{\\ David Stephens}
\textit{\\ Department of Mathematics and Statistics, McGill
University}
\textit{\\ Department of Mathematics, Imperial College London}
\large{\\ Sofia Olhede}
\textit{\\ Department of Statistical Science, University College London}
\vspace{0.2 in}
\end{center}%
%

\begin{changemargin}{.5in}{.5in}
%

\indent
This paper presents hedging strategies for European and exotic options in a
L\'{e}vy market. \ By applying Taylor's theorem, dynamic hedging portfolios
are constructed under different market assumptions, such as the existence of
power jump assets or moment swaps. \ In the case of European options or
baskets of European options, static hedging is implemented. \ It is shown that
perfect hedging can be achieved. \ Delta and gamma hedging strategies are
extended to higher moment hedging by investing in other traded derivatives
depending on the same underlying asset. \ This development is of practical
importance as such other derivatives might be readily available. \ Moment
swaps or power jump assets are not typically liquidly traded. \ It is shown
how minimal variance portfolios can be used to hedge the higher order terms in
a Taylor expansion of the pricing function, investing only in a risk-free bank
account, the underlying asset and potentially variance swaps. \ The numerical
algorithms and performance of the hedging strategies are presented, showing
the practical utility of the derived results.


\noindent%

\end{changemargin}%

\section{Introduction}

This paper provides perfect hedging strategies and minimal variance portfolios
for European and exotic options in a L\'{e}vy market. \ It is well known, see
\cite[p.71]{s00}, that Brownian motion has the Chaotic Representation Property
(CRP), which states that every square integrable random variable adapted to
the filtration generated by a Brownian motion can be represented as a sum of
its mean and an infinite sum of iterated stochastic integrals with respect to
the Brownian motion, with deterministic integrands. \ A consequence of this is
the so-called Predictable Representation Property (PRP) for Brownian motion.
\ The PRP states that every square integrable random variable adapted to the
filtration generated by a Brownian motion can be represented as a sum of its
mean and a stochastic integral with respect to the Brownian motion, where the
integrand is a predictable process. \ The PRP implies the completeness of the
Black-Scholes option pricing model and gives the admissible self-financing
strategy of replicating a contingent claim whose price only depends on the
time to maturity and the current stock price.

\pagestyle{fancy} \fancyhf{} \fancyhead[LE,RO]{\textbf{\thepage}}
\renewcommand{\headrulewidth}{0pt} \fancyhead[CE,CO]{HEDGING STRATEGIES AND
MINIMAL VARIANCE PORTFOLIOS}

Unfortunately, this kind of PRP, where the stochastic integral is with respect
to the underlying process only, is a property which is only possessed by a few
martingales, including Brownian motion, the compensated Poisson process, and
the Az\'{e}ma martingale (see \cite{s03} or \cite{dp99}).\ \ When the
underlying asset is driven by a L\'{e}vy process, perfect hedging using only a
risk-free bank account and the underlying asset is not in general possible.
\ The market is therefore said to be incomplete. \ However, even in this case,
further developments are possible. \ There are two different types of chaos
expansions for L\'{e}vy processes: \cite{i56} proved a Chaotic Representation
Property (CRP) for any square integrable functional of a general L\'{e}vy
process. \ This CRP is written in terms of multiple integrals with respect to
a two-parameter random measure associated with the L\'{e}vy process.
\ \cite{ns00} proved the existence of a new version of the CRP for L\'{e}vy
processes which satisfy some exponential moment conditions. \ This new version
states that every square integrable random variable adapted to the filtration
generated by a L\'{e}vy process can be represented as an infinite sum of
iterated stochastic integrals with respect to the orthogonalised compensated
power jump processes of the underlying L\'{e}vy process. \ The market can be
completed by allowing trades in these processes while risks due to jumps and
fat tails are considered. \ \ In light of the new version of the PRP,
\cite{cns05} suggested that the market should be enlarged with power jump
assets so that perfect hedging could still be implemented. \ \cite{cgns06}
used this completeness to solve the portfolio optimisation problem using the
martingale method. \ Another form of commonly traded financial derivative is
the variance swap which depends functionally on the volatility of the
underlying asset. \ Since variance swaps are already traded commonly in the
over-the-counter (OTC) markets, \cite{s05} suggested trading in moment swaps,
which are a generalisation of variance swaps. \ Based on the CRP derived by
\cite{i56}, \cite{bdlop03} derived a minimal variance portfolio for hedging
contingent claims in a L\'{e}vy market. \ 

Inspired by these papers, we derive practical and implementable hedging
strategies based on the PRP derived from Taylor approximations to the\ option
pricing formulae. \ We apply Taylor's theorem directly to the option pricing
formulae and derive perfect hedging strategies by investing in power jump
assets, moment swaps or some traded derivatives depending on the same
underlying asset. \ The hedging of the higher moments terms in the Taylor
expansion of a contingent claim using other contingent claims in a L\'{e}vy
market is a technique introduced by this paper. \ When these financial
derivatives are not available, we demonstrate how to use the minimal variance
portfolios derived by \cite{bdlop03} to hedge the higher order terms in the
Taylor expansion. \ While we apply Taylor expansions to decompose the pricing
formula into an infinite sum of higher moment terms, \cite{cns05} applied the
It\^{o}'s formula to obtain the PRP of a contingent claim. \ Note that the
It\^{o}'s formula is derived as a result of the elementary Taylor expansion,
see \cite{k02b}. \ In practice, when implementing a hedging strategy
numerically, we have to discretise the time variable. \ Hence, it is more
natural to work directly from Taylor's theorem as this discretisation can be
acknowledged explicitly. \ In fact, the delta and gamma hedges commonly used
by traders in the market, given in Section \ref{SectionDGLit}, are derived
using a Taylor expansion. \ We construct dynamic hedging strategies for
European and exotic options in a L\'{e}vy market. \ Although static hedging is
only applied to European options, exotic options can be decomposed into a
basket of European options so that static hedging can be achieved; in this
case see for example \cite{dek95}. \ It is practically important to be able to
statically hedge since static hedging has several advantages over dynamic
hedging. \ Static hedging is less sensitive to the assumption of zero
transaction costs (both commissions and the cost of paying individuals to
monitor the positions). \ Moreover, dynamic hedging tends to fail when
liquidity dries up or when the market makes large moves, but especially in
these situations effective hedging is needed.

We discuss how hedging can be implemented by applying Taylor's theorem to a
pricing formula. \ We investigate the approximation of the derivatives of the
pricing formula and present the numerical procedures used to construct the
hedging strategies. \ Performance of the hedging is assessed and the
difficulties encountered are discussed. \ Thus, this paper constitutes a
practical development for the hedging of contingent claims in a L\'{e}vy market.

The rest of the paper is arranged as follows: \ Section
\ref{SectionBackground} gives the background about L\'{e}vy processes and the
CRPs in terms of power jump processes and Poisson random measures. \ Section
\ref{SectionHS} gives hedging strategies by investing in variance swaps,
moment swaps or power jump assets and extend the delta and gamma hedging
strategies to higher moment hedging. \ Section \ref{SectionMVP} gives the
minimal variance portfolios in the case where perfect hedging is not possible.
\ Section \ref{SectionDer} gives the approximation procedures of the hedging
strategies and Section \ref{SectionPerform} gives the performance of the
hedging strategies implemented on a set of different types of options as
illustration of the performance of the proposed method. \ An example of static
hedge of a one year European option on real life data is given. \ In Section
\ref{SectionConclusion}, some concluding remarks are provided. \ Proofs and
tables are included in the appendices at the end.

\section{Background\label{SectionBackground}\label{SectionIntroduction}}

In this section, we give a brief introduction to L\'{e}vy processes and the
two versions of the CRP discussed in the introduction. \ We discuss Taylor
expansions which will later be used to derive the hedging portfolios of exotic options.

Let $X=\{X_{t},t\geq0\}$ be a L\'{e}vy process in a complete probability space
$(\Omega,\mathcal{F},P)$ on $\mathbb{R}^{d}$, where $\mathcal{F}$ is the
filtration generated by $X:\mathcal{F}_{t}=\sigma\left\{  X_{t},0\leq s\leq
t\right\}  $, where $\sigma$ denotes the sigma-algebra generated by $X. $ \ A
detailed account of L\'{e}vy process can be found in \cite{s99}. \ Denote the
\textit{left limit process} by $X_{t-}=\lim_{s\rightarrow t,s<t}X_{s},\ t>0,$
and the \textit{jump size} at time $t$ by $\Delta X_{t}=X_{t}-X_{t-}$. \ Let
$\nu$ be the L\'{e}vy measure of $X$. \ In the rest of the paper, we assume
that all L\'{e}vy measures concerned satisfy, for some $\varepsilon>0$ and
$\lambda>0$, $\int_{\left(  -\varepsilon,\varepsilon\right)  ^{c}}\exp\left(
\lambda\left\vert x\right\vert \right)  \nu\left(  \mathrm{d}x\right)
<\infty.$ \ This condition implies that for $i\geq2,$ $\int_{-\infty}%
^{+\infty}\left\vert x\right\vert ^{i}\nu\left(  \mathrm{d}x\right)  <\infty,$
and that the characteristic function $E\left[  \exp\left(  \mathrm{i}%
uX_{t}\right)  \right]  $ is analytic in a neighborhood of 0. \ Denote the
$i$-th \textit{power jump process} by $X_{t}^{(i)}=\sum_{0<s\leq t}(\Delta
X_{s})^{i},\ i\geq2,$ and for completeness let $X_{t}^{(1)}=X_{t}$. \ Clearly
$E[X_{t}]=E[X_{t}^{(1)}]=m_{1}t$, where $m_{1}<\infty$ is a constant and by
\cite[p.32]{p04}, we have%
\begin{equation}
E[X_{t}^{(i)}]=E[\sum_{0<s\leq t}(\Delta X_{s})^{i}]=t\int_{-\infty}^{\infty
}x^{i}\nu(\mathrm{d}x)=m_{i}t<\infty,\ \ \ \mathrm{for}\ i\geq2, \label{mean}%
\end{equation}
thus defining the moments $m_{i}$. \ \cite{ns00} introduced the
\textit{compensated power jump process }(or\textit{\ Teugels martingale}) of
order $i $, $Y^{\left(  i\right)  }=\left\{  Y_{t}^{\left(  i\right)  }%
,t\geq0\right\}  , $ defined by
\begin{equation}
Y_{t}^{(i)}=X_{t}^{(i)}-E[X_{t}^{(i)}]=X_{t}^{(i)}-m_{i}t\ \ \ \mathrm{for}%
\ i=1,2,3,.... \label{CompPower}%
\end{equation}
$Y^{\left(  i\right)  }$ is constructed to have a zero mean. \ It was shown by
\cite[Section 2]{ns00} that there exist constants $a_{i,1},a_{i,2}%
,...,a_{i,i-1}$ such that the processes defined by
\begin{equation}
H_{t}^{(i)}=Y_{t}^{(i)}+a_{i,i-1}Y_{t}^{(i-1)}+\cdots+a_{i,1}Y_{t}^{(1)},
\label{H}%
\end{equation}
for$\ i\geq1$ are a set of pairwise strongly orthogonal
martingales.\ \ \cite{ns00} proved the CRP and PRP in terms of these
orthogonalised compensated power jump processes, $H^{\left(  i\right)  \prime
}$s.

\begin{theorem}
[\textbf{Chaotic Representation Property (CRP)}]Every random variable $F$ in
$L^{2}(\Omega,\mathcal{F})$ has a representation of the form%
\begin{equation}
F=E(F)+\sum_{j=1}^{\infty}\sum_{i_{1},...,i_{j}\geq1}\int_{0}^{\infty}\int
_{0}^{{t_{1}}-}\cdots\int_{0}^{{t_{j-1}}-}f_{(i_{1},...,i_{j})}(t_{1}%
,...,t_{j})\mathrm{d}H_{t_{j}}^{(i_{j})}...\mathrm{d}H_{t_{2}}^{(i_{2}%
)}\mathrm{d}H_{t_{1}}^{(i_{1})}, \label{crp}%
\end{equation}
where the $f_{(i_{1},...,i_{j})}$'s are functions in $L^{2}(\mathbb{R}_{+}%
^{j})$ and $H$'s are defined in equation (\ref{H}).
\end{theorem}

\begin{theorem}
[\textbf{Predictable Representation Property (PRP)}]The CRP stated above
implies that every random variable $F$ in $L^{2}\left(  \Omega,\mathcal{F}%
\right)  $ has a representation of the form%
\begin{equation}
F=E\left[  F\right]  +\sum_{i=1}^{\infty}\int_{0}^{\infty}\phi_{s}^{\left(
i\right)  }\mathrm{d}H_{s}^{\left(  i\right)  }, \label{prp}%
\end{equation}
where $\phi_{s}^{\left(  i\right)  }$'s are predictable, that is, they are
$\mathcal{F}_{s-}$-measurable.
\end{theorem}

In contrasts, \cite{i56} proved a chaos expansion for general L\'{e}vy
processes in terms of multiple integrals with respect to the compensated
Poisson random measure. \ One may convert the representation to one involving
iterated integrals (see \cite{l04}). \ The stochastic integrals are in terms
of both Brownian motion, $W$, and the compensated Poisson measure $\tilde
{N},$
\begin{equation}
\tilde{N}\left(  \mathrm{d}t,\mathrm{d}x\right)  =N\left(  \mathrm{d}%
t,\mathrm{d}x\right)  -\nu\left(  \mathrm{d}x\right)  \mathrm{d}t,
\label{cprm}%
\end{equation}
where $\nu\left(  \mathrm{d}x\right)  $ is the L\'{e}vy measure of the
underlying L\'{e}vy process, and\
\[
N\left(  B\right)  =\#\left\{  t:\left(  t,\Delta X_{t}\right)  \in B\right\}
,B\in\mathcal{B}\left(  \left[  0,T\right]  \times\mathbb{R}_{0}\right)  ,
\]
is the Poisson random measure of the L\'{e}vy process where $\mathcal{B}%
\left(  \left[  0,T\right]  \times\mathbb{R}_{0}\right)  $ is the
\textit{Borel} $\sigma$\textit{-algebra} of $\left[  0,T\right]
\times\mathbb{R}_{0}$ and $\mathbb{R}_{0}=\mathbb{R-}\left\{  0\right\}
$.\ \ The compensator of $N\left(  \mathrm{d}t,\mathrm{d}x\right)  $ is known
to be $E\left[  N\left(  \mathrm{d}t,\mathrm{d}x\right)  \right]  =\nu\left(
\mathrm{d}x\right)  \mathrm{d}t.$ \ \cite{bdlop03} derived relations between
the two chaos expansions. \ When the underlying L\'{e}vy process is a pure
jump process, the compensated power jump process defined in (\ref{CompPower})
satisfies%
\begin{equation}
Y_{t}^{\left(  i\right)  }=\int_{0}^{t}\int_{\mathbb{R}}x^{i}\tilde{N}\left(
\mathrm{d}s,\mathrm{d}x\right)  ,\text{ \ \ }0\leq t\leq T,\ i=1,2,....
\label{YN}%
\end{equation}
This relationship is very important in the development of the chaotic
representation of L\'{e}vy processes. \ Since the introduction of the chaos
expansion by \cite{i56}, the development of representations in the literature
has been focused on expansions with respect to the Poisson random measure.
\ Unfortunately, we cannot trade in the Poisson random measure. \ Note that
trading in a finite set of power jump assets is possible because the $i$-th
power jump asset contains information of the $i$-th moment of the L\'{e}vy
process, given that $i$ is finite. \ Therefore, it is theoretically possible
to construct a financial product which contains information of the $i$-th
moment of the underlying process. \ For example, if we want to hedge the risk
introduced by the variance of the underlying process, we can trade in the
variance swaps or the second power jump asset. \ However, the Poisson random
measure contains \textbf{all} the information of the moments up to infinity
and hence it is not clear how to construct such a financial product unless
information of all the higher moments are obtained. \ This limits the
application of the CRP in terms of Poisson random measures and also the
application of L\'{e}vy processes in finance. \ The equation (\ref{YN}) links
the two important expansions together and hence the results derived for
expansions in terms of Poisson random measures can be applied to the
expansions in terms of power jump processes.

To unify notation, \cite{bdlop03} defined the following notation:
\[
U_{1}=\left[  0,T\right]  ,U_{2}=\left[  0,T\right]  \times\mathbb{R}%
,\mathrm{d}Q_{1}\left(  \cdot\right)  =\mathrm{d}W\left(  \cdot\right)
,Q_{2}\left(  \cdot\right)  =\tilde{N}\left(  \cdot,\cdot\right)  ,
\]%
\[
\int_{U_{1}}g\left(  u^{\left(  1\right)  }\right)  Q_{1}\left(
\mathrm{d}u^{\left(  1\right)  }\right)  =\int_{0}^{t}g\left(  s\right)
W\left(  \mathrm{d}s\right)  ,\int_{U_{2}}g\left(  u^{\left(  2\right)
}\right)  Q_{2}\left(  \mathrm{d}u^{\left(  2\right)  }\right)  =\int_{0}%
^{t}\int_{\mathbb{R}}g\left(  s,x\right)  \tilde{N}\left(  \mathrm{d}%
s,\mathrm{d}x\right)  .
\]
The CRP in terms of Brownian motion and Poisson random measures is given by:

\begin{theorem}
[Chaos expansion for general L\'{e}vy process by \cite{i56}]\label{TheoremIto}%
Let $F$ be a square integrable random variable adapted to the underlying
L\'{e}vy process, $X$. \ We have%
\begin{equation}
F=E\left[  F\right]  +\sum_{n=1}^{\infty}\sum_{j_{1},...,j_{n}=1,2}\tilde
{J}_{n}\left(  g_{n}^{\left(  j_{1},...,j_{n}\right)  }\right)  , \label{gprm}%
\end{equation}
for a unique sequence $g_{n}^{\left(  j_{1},...,j_{n}\right)  }$
($j_{1},...,j_{n}=1,2;\ n=1,2,...$) of deterministic functions in the
corresponding $L_{2}$-space, $L_{2}\left(  G_{n}\right)  ,$ where
\begin{equation}
G_{n}=\left\{  \left(  u_{1}^{\left(  j_{1}\right)  },...,u_{n}^{\left(
j_{n}\right)  }\right)  \in\Pi_{i=1}^{n}U_{j_{i}}:0\leq t_{1}\leq\cdots\leq
t_{n}\leq T\right\}  \label{gn}%
\end{equation}
with $u^{\left(  j_{i}\right)  }=t$ if $j_{i}=1$, and $u^{\left(
j_{i}\right)  }=\left(  t,x\right)  $ if $j_{i}=2,$ and
\begin{align*}
&  \tilde{J}_{n}\left(  g_{n}^{\left(  j_{1},...,j_{n}\right)  }\right) \\
&  \ =\int_{\Pi_{i=1}^{n}U_{j_{i}}}g_{n}^{\left(  j_{1},...,j_{n}\right)
}\left(  u_{1}^{\left(  j_{1}\right)  },...,u_{n}^{\left(  j_{n}\right)
}\right)  1_{G_{n}}\left(  u_{1}^{\left(  j_{1}\right)  },...,u_{n}^{\left(
j_{n}\right)  }\right)  Q_{j_{1}}\left(  \mathrm{d}u_{1}^{\left(
j_{1}\right)  }\right)  \cdots Q_{j_{n}}\left(  \mathrm{d}u_{n}^{\left(
j_{n}\right)  }\right)  .
\end{align*}

\end{theorem}

So far we have given the theoretical results on the chaotic representations.
\ We now discuss their financial applications. \ Under the Black-Scholes
model, the PRP of Brownian motions allows perfect hedging of European options.
\ Unfortunately, the derivation of hedging strategies of options in an
incomplete market\ is not as simple and has been the focus of considerable
study in the literature, see for example \cite{cgm01}, \cite{hkcflv05} and
\cite{ctv05}. \ In this paper, by extending the ideas of \cite{cns05},
\cite{s05} and \cite{bdlop03}, we derive and implement some hedging strategies
for European and exotic options. \ Numerical procedures are provided and
performance of the hedging strategies is discussed.

As mentioned above, the PRP is useful in option hedging. \ For option pricing
functions which are infinitely differentiable in the stock price, we can
simply apply the It\^{o}'s formula to obtain such a predictable
representation. \ Assuming power jump assets are traded in the market,
\cite{cns05} derived a self-financing replicating portfolio for a contingent
claim whose payoff function only depends on the stock price at maturity.
\ Their hedging formula is derived from the It\^{o}'s formula and given in
terms of an infinite sum of stochastic integrals.\ \ In this paper, we use a
different approach to determine a self-financing replicating portfolio, which,
in some cases, can be used in both static and dynamic hedging with a flexible
$\Delta t,$ where $\Delta t$ denotes the time change during the hedging
period. \ We discuss this in more detail in Section \ref{SectionHS}. \ We
apply Taylor's theorem directly to the option pricing formulae to obtain the
hedging portfolios. \ In the literature, the results on option hedging using
CRP of L\'{e}vy processes, has previously focused on the theoretical aspects
of the problem, see, for example, \cite{cns05} and \cite{l04}. \ We aim to
investigate the problem from a practical point of view by providing methods to
obtain the hedging portfolios explicitly using numerical methods and shall
discuss the difficulties encountered. \ Our approach can be applied to a range
of exotic options in the case of dynamic hedging, for example, barrier
options, lookback options and Asian options.

\section{The perfect hedging strategies\label{SectionHS}}

In this section, we derive hedging strategies using Taylor's theorem.
\ Firstly, we specify the model of the underlying asset, $S_{t}$. \ Following
\cite[Theorem 3]{cns05}, we assume
\begin{equation}
\mathrm{d}S_{t}=bS_{t-}\mathrm{d}t+S_{t-}\mathrm{d}X_{t}, \label{S}%
\end{equation}
where $X=\left\{  X_{t},t\geq0\right\}  $ is a general L\'{e}vy process. \ Let
the risk-free bank account be $B_{t}=\exp\left(  rt\right)  ,$ where $r$ is
the continuously compounded risk-free rate. \ Let $F\left(  t,x\right)  $ be
the option pricing function at time $t<T$ and stock price equal to $x,$ where
$T$ is the maturity of the option. \ Let $D_{1}^{i}F\left(  t,x\right)  $ be
the $i$-th derivative of $F\left(  t,x\right)  $ with respect to the first
variable (time), and $D_{2}^{i}F\left(  t,x\right)  $ be the $i$-th derivative
of $F\left(  t,x\right)  $ with respect to the second variable (stock price).
\ Suppose $F\left(  t,x\right)  $ is continuous and infinitely differentiable
in the second variable and satisfies
\[
\sup_{x<K,t\leq t_{0}}\sum_{n=2}^{\infty}\left\vert D_{2}^{n}F\left(
t,x\right)  \right\vert R^{n}<\infty\qquad\text{for all }K,R>0,t_{0}>0.
\]
Let $\Delta t$ be the time change during the hedging period and $\Delta
S_{t}=S_{t+\Delta t}-S_{t}$. \ Applying Taylor's theorem twice to the option
pricing formula, $F\left(  t,S_{t}\right)  $, we obtain%
\begin{equation}
F\left(  t+\Delta t,S_{t}+\Delta S_{t}\right)  -F\left(  t,S_{t}\right)
=\sum_{i=1}^{\infty}\frac{D_{2}^{i}F\left(  t+\Delta t,S_{t}\right)  }%
{i!}\left(  \Delta S_{t}\right)  ^{i}+\sum_{i=1}^{\infty}\frac{D_{1}%
^{i}F\left(  t,S_{t}\right)  }{i!}\left(  \Delta t\right)  ^{i}, \label{f1}%
\end{equation}
which is true as long as $D_{2}^{i}F\left(  t+\Delta t,S_{t}\right)  $ and
$D_{1}^{i}F\left(  t,S_{t}\right)  $ exist for $i=1,2,3,....$ \ Note that it
is not necessary to apply the multivariate Taylor's theorem since the value of
$\Delta t$ is known at time $t$. $\ $Let $M^{\left(  q\right)  }\left(
t,x\right)  $ be the price of a financial derivative such that $M^{\left(
q\right)  }\left(  0,S_{0}\right)  =F\left(  0,S_{0}\right)  $ and%
\begin{equation}
M^{\left(  q\right)  }\left(  t+\Delta t,S_{t}+\Delta S_{t}\right)
-M^{\left(  q\right)  }\left(  t,S_{t}\right)  =\sum_{i=1}^{q}\frac{D_{2}%
^{i}F\left(  t+\Delta t,S_{t}\right)  }{i!}\left(  \Delta S_{t}\right)
^{i}+\sum_{i=1}^{\infty}\frac{D_{1}^{i}F\left(  t,S_{t}\right)  }{i!}\left(
\Delta t\right)  ^{i}, \label{f2}%
\end{equation}
where $q$ is a positive integer. \ Therefore, we have $\lim_{q\rightarrow
\infty}M^{\left(  q\right)  }\left(  T,S_{T}\right)  =F\left(  T,S_{T}\right)
,$ that is, the value of the financial derivative $M^{\left(  q\right)  }$
converges to $F$ as $q$ goes to infinity. \ Our aim is to construct a
self-financing hedging portfolio for $M^{\left(  q\right)  }.$ \ Note that the
hedging error at time $\Delta t$,
\begin{align*}
\left[  F\left(  t+\Delta t,S_{t}+\Delta S_{t}\right)  -F\left(
t,S_{t}\right)  \right]   &  -\left[  M^{\left(  q\right)  }\left(  t+\Delta
t,S_{t}+\Delta S_{t}\right)  -M^{\left(  q\right)  }\left(  t,S_{t}\right)
\right] \\
&  \qquad=\sum_{i=q+1}^{\infty}\frac{D_{2}^{i}F\left(  t+\Delta t,S_{t}%
\right)  }{i!}\left(  \Delta S_{t}\right)  ^{i},
\end{align*}
can be approximated using standard techniques in calculating the remainder
terms in a Taylor expansion. \ Let $\mathcal{P}_{t}^{\left(  i\right)  }$ be
the value of a basket of financial derivatives such as the risk-free bank
account, the underlying stock, variance swaps, moment swaps, power jump assets
or other financial derivatives depending on the same underlying stock such
that $\left(  \Delta S_{t}\right)  ^{i}=\Delta\mathcal{P}_{t}^{\left(
i\right)  }=\mathcal{P}_{t+\Delta t}^{\left(  i\right)  }-\mathcal{P}%
_{t}^{\left(  i\right)  }.$ \ Note that $\mathcal{P}_{t}^{\left(  i\right)  }$
is a basket of assets that would not lead to arbitrage opportunities. \ We
will show later how to construct such a basket of tradable assets.
\ Therefore, we have%
\begin{align}
M^{\left(  q\right)  }\left(  t+\Delta t,S_{t}+\Delta S_{t}\right)
-M^{\left(  q\right)  }\left(  t,S_{t}\right)   &  =\sum_{i=1}^{\infty}%
\frac{D_{1}^{i}F\left(  t,S_{t}\right)  }{i!}\left(  \Delta t\right)
^{i}+D_{2}^{1}F\left(  t+\Delta t,S_{t}\right)  \Delta S_{t}\nonumber\\
&  +\sum_{i=2}^{q}\frac{D_{2}^{i}F\left(  t+\Delta t,S_{t}\right)  }{i!}%
\Delta\mathcal{P}_{t}^{\left(  i\right)  }. \label{hedgePort}%
\end{align}
The self-financing portfolio to hedge $M^{\left(  q\right)  }\left(  t+\Delta
t,S_{t}+\Delta S_{t}\right)  -M^{\left(  q\right)  }\left(  t,S_{t}\right)  $
is then

\begin{itemize}
\item[(i)] Invest
\begin{equation}
\frac{1}{\left(  \exp\left(  r\Delta t\right)  -1\right)  }\sum_{i=1}^{\infty
}D_{1}^{i}F\left(  t,S_{t}\right)  \left(  \Delta t\right)  ^{i}/i! \label{P}%
\end{equation}
in a riskless bank account such that at time $t+\Delta t$, the deposit is
worth
\[
\frac{1}{\left(  \exp\left(  r\Delta t\right)  -1\right)  }\sum_{i=1}^{\infty
}D_{1}^{i}F\left(  t,S_{t}\right)  \left(  \Delta t\right)  ^{i}\exp\left(
r\Delta t\right)  /i!
\]
and the change of value of the investment is $\sum_{i=1}^{\infty}\frac
{D_{1}^{i}F\left(  t,S_{t}\right)  }{i!}\left(  \Delta t\right)  ^{i}$.

\item[(ii)] Invest $D_{2}^{1}F\left(  t+\Delta t,S_{t}\right)  $ in the
underlying stock;

\item[(iii)] Invest $\dfrac{D_{2}^{i}F\left(  t+\Delta t,S_{t}\right)  }{i!}$
in $\mathcal{P}_{t}^{\left(  i\right)  }$ for $i=2,3,...,q$.
\end{itemize}

In real life application, we have to find a reasonable value for $q$ and we
discuss methods of choosing $q$ in Section \ref{SectionPerform}. \ Note that
the approximation in (\ref{f2}) requires the existence of $D_{1}^{i}F\left(
t,S_{t}\right)  $ for $i=1,2,3,...$ and $D_{2}^{i}F\left(  t+\Delta
t,S_{t}\right)  $ only for $i=1,2,3,...,q.$ \ The value of $q$ determines how
many different financial derivatives needed to hedge the option up to a
pre-specified level of accuracy. \ If $q=1,$ we only need to hedge the
deterministic term
that appears as the first term in equation (\ref{hedgePort}) by investing in a
risk-free bank account, and the term $D_{2}^{1}F\left(  t+\Delta
t,S_{t}\right)  \Delta S_{t}$ by investing in the underlying stock, which is a
simple extension to the delta hedging discussed in Section \ref{SectionDGLit}.
\ If $q=2,$ we can hedge by investing in a risk-free bank account, the
underlying stock and the variance swaps currently traded in the market, which
is discussed in Section \ref{SectionVS}. \ If $q\geq3,$ we can consider
perfect hedging in three cases: \ (a) trading in moment swaps, discussed in
Section \ref{SectionMS},\ (b) trading in power jump assets, discussed in
Section \ref{SectionHPJ} and (c) trading in some financial derivatives
depending on the same underlying assets, discussed in Section \ref{SectionExt}%
. \ Note that (a) and (b) are not liquidly traded in the market while (c)
might be more readily available. \ If all of these financial derivatives are
not available for trading, we can employ the minimal variance portfolios
derived in Section \ref{SectionMVP}. \ 

The approximation in (\ref{f2}) can be used in both static and dynamic hedging
for European options by just changing $\Delta t$. \ The reason why static
hedging may not be applicable to exotic options is because if during the
hedging period, $\Delta t$, the value of the $S_{t+\Delta s}$, where $\Delta
s<\Delta t$ is explicitly occurring in the formulae, then this must be used in
the calculation of the option price. \ In this case, we have to apply Taylor's
theorem with respect to both $\Delta S_{t}=\left(  S_{t+\Delta t}%
-S_{t}\right)  $ and $\left(  S_{t+\Delta s}-S_{t}\right)  $. \ In the case of
dynamic hedging, we can assume that the minimum time period for a change of
value of $S$ to take place is equal to $\Delta t\,$, the hedging period$.$
Although static hedging can only be applied to European options, some exotic
options can be decomposed into a basket of European options such that static
hedging can still be achieved, see for example \cite{dek95}. \ In Section
\ref{SectionPerform}, we show the approximation results for both static
hedging ($\Delta t$ equals to 3 months) and dynamic hedging ($\Delta t$ equals
to 5 minutes) for European options and dynamic hedging for barrier options.
\ The advantage of static hedging over dynamic hedging is that in real life,
transaction costs and bid-ask spreads of option prices are not negligible.
\ The replicating portfolio is not truly self-financing since extra investment
must be made to pay for these additional costs. \ Hence, it is preferable to
hedge statically rather than dynamically as the costs involved will be less
and constant rebalancing is not required. \ In the literature and in practice,
it is common to assume that $\Delta S_{t}$ is very small such that the
approximation in (\ref{f2}) can be truncated without loss of accuracy; this is
the main assumption behind the delta and gamma hedges commonly used by traders
in the market. \ However, in real life, the price of every traded asset in the
market moves by a tick size, such as 0.5 or 1. \ After a very short period of
time, the price of the traded asset either stays unchanged or moves by a
multiple of the tick size. \ Hence, the assumption of $\Delta S_{t}$ being
very small in hedging is not sufficiently accurate. \ It would not in general
be reasonable to assume that $\Delta S_{t}$ is small when modelling $S$ as a
process with jumps. \ Thus, we consider $\Delta S_{t}\geq1$ for both static
and dynamic hedging in our simulation analysis in Section \ref{SectionPerform}.

\subsection{Hedging instruments\label{SectionVSMS}\label{SectionPJA}}

In this section, we consider the use of moment swaps (including variance
swaps) and power jump assets in our hedging strategies. \ Recall the PRP for
L\'{e}vy processes involves stochastic integrals with respect to power jump
processes, which are related to the higher moments of the underlying L\'{e}vy
process. \ In equation (\ref{f2}), they are represented through the terms
$\frac{D_{2}^{i}F\left(  t+\Delta t,S_{t}\right)  }{i!}\left(  \Delta
S_{t}\right)  ^{i}.$ \ To hedge these terms, we need to invest in some
financial derivatives related to these higher moments. \ We show how moment
swaps introduced by \cite{s05} and power jump assets by \cite{cns05} can be
used to construct $\mathcal{P}_{t}^{\left(  i\right)  }$ used in the hedging
portfolio given in (\ref{hedgePort}). \ Variance swaps, introduced by
\cite{ddkz99}, are commonly traded over-the-counter (OTC) derivatives.
\ \cite{s05} generalised variance swaps to moment swaps, which are not
liquidly traded in the market. \ \cite{wfv06} gave a detailed discussion on
volatility swaps. \ There are two common contractual definitions of returns of
stock price. \ Let $\left\{  s_{1},s_{2},...,s_{n}\right\}  $ be the sampling
points of the contract, where the $s$'s are equally spaced with length $\Delta
s$. \ The \textit{actual return }and the \textit{log return} are defined to be%
\begin{equation}
R_{\text{actual},i}=\left(  S_{s_{i+1}}-S_{s_{i}}\right)  /S_{s_{i}}\text{
\ \ and \ \ }R_{\text{log},i}=\log\left(  S_{s_{i+1}}/S_{s_{i}}\right)  .
\label{R}%
\end{equation}
The annualised {realised} variance, $\sigma_{\text{{realised}}}^{2}$, is
defined by $\sigma_{\text{{realised}}}^{2}=\frac{1}{\Delta s\left(
n-2\right)  }\sum_{i=1}^{n-1}R_{i}^{2}$ where $R_{i}$ is either the actual
return or log return of the stock price. \ In the case of log return,
$R_{i}=R_{\text{log},i}$, \cite{s05} generalised variance swaps to moment
swaps. \ The annualised {realised} $k$-th moment is defined by
$M_{\text{{realised}}}^{\left(  k\right)  }=\frac{1}{\Delta s\left(
n-2\right)  }\sum_{i=1}^{n-1}R_{i}^{k}.$ \ This definition can be easily
extended to the case where $R_{i}=R_{\text{actual},i}.$ We can now give the
definition of the $k$-th moment swap.

\begin{definition}
\label{DefinitionMS}A $k$-th moment swap is a forward contract on annualised
{realised} $k$-th moment, $M_{\text{{realised}}}^{\left(  k\right)  }$. \ Its
payoff to the holder at expiration is equal to $\left(  M_{\text{{realised}}%
}^{\left(  k\right)  }-M_{\text{strike}}^{\left(  k\right)  }\right)  N,$
where $M_{\text{{realised}}}^{\left(  k\right)  }$ is the {realised} $k$-th
moment (quoted in annual terms) over the life of the contract,
$M_{\text{strike}}^{\left(  k\right)  }$ is the pre-defined delivery price for
the $k$-th moment, and $N$ is the notional amount of the swap. The annualised
{realised} $k$-th moment is calculated based on the pre-specified set of
sampling points over the period, $\left\{  s_{1},s_{2},...,s_{n}\right\}  $. \ 
\end{definition}

\cite{cns05} suggested enlarging the L\'{e}vy market with power jump assets,
where the $i$-th power jump asset is defined by%
\begin{equation}
T_{t}^{\left(  i\right)  }=\exp\left(  rt\right)  Y_{t}^{\left(  i\right)
},\ \ \ i\geq2, \label{T}%
\end{equation}
and $Y_{t}^{\left(  i\right)  }$ is defined in (\ref{CompPower}). \ The
authors derived the dynamic hedging portfolio trading in these assets using
the It\^{o}'s formula and noted that the 2nd power jump process is related to
the {realised} variance. \ However, the 2nd power jump asset is not the same
as a variance swap and we consider their usages separately in Section
\ref{SectionHS2}.

\subsection{Hedging strategies\label{SectionHS2}}

In the last section, we introduce two different kinds of financial derivatives
involving higher moments, namely, the moment swaps and the power jump assets.
\ In this section, we explain how to use them to construct the basket of
financial derivatives, $\mathcal{P}_{t}^{\left(  i\right)  }$, in order to
hedge the terms in equation (\ref{f2}). \ We also discuss the delta and gamma
hedges in the literature and we extend them in order to obtain perfect hedging
by trading in certain financial derivatives depending on the same underlying
asset, which may be available in the market.

\subsubsection{Hedging with variance swaps\label{SectionVS}}

To hedge the term $\left(  \Delta S_{t}\right)  ^{2}$ in equation (\ref{f2}),
we construct $\mathcal{P}_{t}^{\left(  2\right)  }$ which invest in a
risk-free bank account and variance swaps. \ If $\Delta t$ is negligible
compared to $\Delta S_{t}$, from (\ref{S}), we have \
\begin{equation}
\left(  \Delta S_{t}\right)  ^{2}=S_{t}^{2}\left(  \Delta X_{t}\right)  ^{2}.
\label{SX}%
\end{equation}
We cannot use the variance swaps using log return, $R_{\text{log},i}$ defined
in (\ref{R}) to hedge. We have
\[
\left[  \log\left(  \frac{S_{t+\Delta t}}{S_{t}}\right)  \right]  ^{2}
\bumpeq\left[  \log\left(  1+\Delta X_{t}\right)  \right]  ^{2}%
\]
since we assume $\Delta t$ to be negligible. From (\ref{SX}), we need $\left(
\Delta X_{t}\right)  ^{2}$ rather than $\left[  \log\left(  1+\Delta
X_{t}\right)  \right]  ^{2}$ to hedge, therefore the variance swaps using log
returns are not useful in this case. \ Even if we use the model $S_{t+\Delta
t}=S_{t}\exp\left(  \Delta X_{t}\right)  $ such that $\log\left(  S_{t+\Delta
t}/S_{t}\right)  =\Delta X_{t},$ we then have $\left(  \Delta S_{t}\right)
^{2}=\left(  S_{t+\Delta t}-S_{t}\right)  ^{2}=S_{t}^{2}\left[  \exp\left(
\Delta X_{t}\right)  -1\right]  ^{2},$ which still can not be hedged by the
variance swaps using log returns. \ Therefore, in our case where we apply
Taylor's theorem with respect to $\Delta S_{t}$, we should invest in the
variance swaps using absolute returns, $R_{\text{actual},i}$, as defined in
(\ref{R}). \ 

Recall in Section \ref{SectionVSMS} that there is a set of sampling points,
$\left\{  s_{1},s_{2},...,s_{n}\right\}  $, for each contract. \ We invest in
the variance swap at time $t$ where the last two sampling points are equal to
$t$ and $t+\Delta t$: $s_{n-1}=t$ and $s_{n}=t+\Delta t$ and maturity equal to
$t+\Delta t$. \ At maturity, we receive the payoff $\sigma_{\text{{realised}}%
}^{2}-\sigma_{\text{strike}}^{2}$, where
\[
\sigma_{\text{{realised}}}^{2}=\frac{1}{\Delta s\left(  n-2\right)  }%
\sum_{i=1}^{n-1}\left(  \frac{S_{t_{i+1}}-S_{t_{i}}}{S_{t_{i}}}\right)
^{2}=\frac{1}{\Delta s\left(  n-2\right)  }\left[  \left(  \frac{\Delta S_{t}%
}{S_{t}}\right)  ^{2}+\sum_{i=1}^{n-2}\left(  \frac{S_{t_{i+1}}-S_{t_{i}}%
}{S_{t_{i}}}\right)  ^{2}\right]
\]
and the value of $\sum\limits_{i=1}^{n-2}\left[  (S_{t_{i+1}}-S_{t_{i}%
})/S_{t_{i}}\right]  ^{2}$ is known at time $t$. In the following, we give the
hedging strategy to hedge the term
\begin{equation}
Q_{2}=\frac{D_{2}^{2}F\left(  t+\Delta t,S_{t}\right)  }{2}\left(  \Delta
S_{t}\right)  ^{2}=C_{2}\left(  \Delta S_{t}\right)  ^{2} \label{eq:D2term}%
\end{equation}
in equation (\ref{f2}) by constructing $\mathcal{P}_{t}^{\left(  2\right)  }$.

\begin{proposition}
\label{PropositionHVS} To hedge the term $Q_{2}$ in equation (\ref{eq:D2term})
we invest in $C_{2}$ units of $\mathcal{P}_{t}^{\left(  2\right)  }$ at time
$t$, consisting of $\Delta s\left(  n-2\right)  S_{t}^{2}$ units of the
variance swap with sampling points $\left\{  ...,s_{n-1}=t,s_{n}=t+\Delta
t\right\}  \,$, maturity $t+\Delta t$, strike $\sigma_{\text{strike}}^{2}$
and
\[
\frac{S_{t}^{2}\Delta s\left(  n-2\right)  }{\left[  \exp\left(  r\Delta
t\right)  -1\right]  }\left[  \sigma_{\text{strike}}^{2}-\frac{1}{\Delta
s\left(  n-2\right)  }\sum_{i=1}^{n-2}\left(  \frac{S_{t_{i+1}}-S_{t_{i}}%
}{S_{t_{i}}}\right)  ^{2}\right]  +\frac{P_{V}\Delta s\left(  n-2\right)
S_{t}^{2}}{\left[  \exp\left(  r\Delta t\right)  -1\right]  }%
\]
units of cash in a risk-free bank account, where $P_{V}$ is the price of one
unit of the variance swap.
\end{proposition}

\begin{proof}
Let
\begin{equation}
\overline{S}_{n,2}=\frac{1}{\Delta s\left(  n-2\right)  }\sum_{i=1}%
^{n-2}\left(  \frac{S_{t_{i+1}}-S_{t_{i}}}{S_{t_{i}}}\right)  ^{2}=\frac
{1}{\Delta s\left(  n-2\right)  }\widetilde{S}_{n,2}. \label{S2bar}%
\end{equation}
The initial investment at time $t$ equals the price of the variance swap plus
the deposit into the risk-free bank account, which is equal to
\[
C_{2}\Delta s\left(  n-2\right)  S_{t}^{2}P_{V}\left[  1+\frac{1}{e^{r\Delta
t}-1}\right]  \newline+\frac{C_{2}S_{t}^{2}\Delta s\left(  n-2\right)
}{\left[  \exp\left(  r\Delta t\right)  -1\right]  }\left[  \sigma
_{\text{strike}}^{2}-\overline{S}_{n,2}\right]  .
\]
At maturity, the portfolio is worth%
\[
\ C_{2}\left(  \Delta S_{t}\right)  ^{2}+C_{2}S_{t}^{2}\Delta s\left(
n-2\right)  \left[  \sigma_{\text{strike}}^{2}-\overline{S}_{n,2}\right]
/\left[  \left(  e^{r\Delta t}-1\right)  \right]  +C_{2}P_{V}\frac{\Delta
s\left(  n-2\right)  S_{t}^{2}e^{r\Delta t}}{e^{r\Delta t}-1}.
\]
Hence, the change of value of the hedging portfolio is equal to
\[
C_{2}\left(  \Delta S_{t}\right)  ^{2}+C_{2}\Delta s\left(  n-2\right)
S_{t}^{2}P_{V}\left[  \frac{e^{r\Delta t}}{e^{r\Delta t}-1}-1-\frac
{1}{e^{r\Delta t}-1}\right]  =C_{2}\left(  \Delta S_{t}\right)  ^{2},
\]
as desired.
\end{proof}

\subsubsection{Hedging with moment swaps\label{SectionMS}}

In the last section, we explained how to hedge the term $Q_{2}$ in equation
(\ref{eq:D2term}) using variance swaps. \ The idea can be extended easily to
moment swaps to hedge the term
\[
Q_{i}=\frac{D_{2}^{i}F\left(  t+\Delta t,S_{t}\right)  }{i!}\left(  \Delta
S_{t}\right)  ^{i}=C_{i}\left(  \Delta S_{t}\right)  ^{i}%
\]
for $i=3,4,5,...$, which can be done by investing in the $i$-th moment swap at
time $t$ with sampling points $s_{n-1}=t$ and $s_{n}=t+\Delta t$ and maturity
equal to $t+\Delta t$. \ At maturity, we receive the payoff
$M_{\text{{realised}}}^{\left(  i\right)  }-M_{\text{strike}}^{\left(
i\right)  }$, where
\[
M_{\text{{realised}}}^{\left(  i\right)  }=\frac{1}{\Delta s\left(
n-2\right)  }\left[  \left(  \frac{\Delta S_{t}}{S_{t}}\right)  ^{i}%
+\sum_{i=1}^{n-2}\left(  \frac{S_{t_{i+1}}-S_{t_{i}}}{S_{t_{i}}}\right)
^{i}\right]  =\frac{1}{\Delta s\left(  n-2\right)  }\left[  \left(
\frac{\Delta S_{t}}{S_{t}}\right)  ^{i}+\widetilde{S}_{n,i}\right]  ,
\]
and the value of $\widetilde{S}_{n,i}$ is known at time $t.$ In the following,
we give the hedging strategy to hedge the term $Q_{i}$ by constructing
$\mathcal{P}_{t}^{\left(  i\right)  }$.

\begin{proposition}
To hedge the terms $Q_{i}$ we invest in $C_{i}$ units of $\mathcal{P}%
_{t}^{\left(  i\right)  }$ at time $t$, consisting of $\Delta s\left(
n-2\right)  S_{t}^{i}$ units of the $i$-th moment swap with sampling points
$\left\{  ...,s_{n-1}=t,s_{n}=t+\Delta t\right\}  \,$, maturity $t+\Delta t$
and strike $M_{\text{strike}}^{\left(  i\right)  }$, and
\[
\frac{S_{t}^{i}\Delta s\left(  n-2\right)  }{\left[  \exp\left(  r\Delta
t\right)  -1\right]  }\left[  M_{\text{strike}}^{\left(  i\right)  }-\frac
{1}{\Delta s\left(  n-2\right)  }\widetilde{S}_{n,i}\right]  +\frac{\Delta
s\left(  n-2\right)  S_{t}^{i}P_{M}}{\left[  \exp\left(  r\Delta t\right)
-1\right]  }%
\]
units of cash in a risk-free bank account where $P_{M}$ is the price of one
unit of the moment swap.
\end{proposition}

\begin{proof}
The proof follows in the same fashion as for Proposition \ref{PropositionHVS}.

\end{proof}

\subsubsection{Hedging with power jump processes of higher
orders\label{SectionHPJ}}

In the last two sections, we discuss how to hedge the term $\sum_{i=1}%
^{q}Q_{i}$ for$\ q\geq2$ using variance swaps and moment swaps. \ We can
instead hedge using power jump assets, discussed in Section \ref{SectionPJA},
if we allow trading of them. \ \ Using It\^{o}'s formula, see \cite[Section
2.3]{cns05}, equation (\ref{S}) has the solution%
\begin{equation}
S_{t}=S_{0}\exp\left(  X_{t}+\left(  b-\sigma^{2}/2\right)  t\right)
\prod_{0<s\leq t}\left(  1+\Delta L_{s}\right)  \exp\left(  -\Delta
L_{s}\right)  , \label{Ss}%
\end{equation}
where $b$ is defined in (\ref{S}), $\sigma^{2}$ is the Brownian variance
parameter and $L$ is the pure jump part of the L\'{e}vy process $X$, see
\cite[Section 2]{cns05} for details. \ In the following, we consider the
simplified case where there is at most one jump of $X$ between $t$ and
$t+\Delta t$, and the general case where there can be infinite number of
jumps. \ Note that the latter case might not be realistic because in reality,
we only observe a discrete series of the underlying stock $S$, while the power
jump processes of the L\'{e}vy process with infinite activity are not
observable. \ Therefore, it appears to be more practical to consider trading
in moment swaps rather than power jump processes. \ We consider both assets
for completeness and theoretical interest.

\paragraph{The simplified case}

\ If $\Delta t$ is negligible compared to $\Delta S_{t}$, from (\ref{S}),
(\ref{T}) and assuming there is at most one jump of $X$ between $t$ and
$t+\Delta t$,\ we have%
\begin{equation}
\left(  \Delta S_{t}\right)  ^{i}=S_{t}^{i}\left[  \exp\left(  -r\left(
t+\Delta t\right)  \right)  T_{t+\Delta t}^{\left(  i\right)  }-\exp\left(
-rt\right)  T_{t}^{\left(  i\right)  }+m_{i}\Delta t\right]  . \label{PJAD}%
\end{equation}
\ Therefore, we can hedge the term $Q_{i}$ by constructing $\mathcal{P}%
_{t}^{\left(  i\right)  }$:

\begin{proposition}
\label{Proposition3}If $\Delta t$ is negligible compared to $\Delta S_{t}$, to
hedge $Q_{i},$ we invest in $C_{i}$ units of $\mathcal{P}_{t}^{\left(
i\right)  },$ consisting of $S_{t}^{i}\exp\left(  -r\left(  t+\Delta t\right)
\right)  $ units of $T_{t}^{\left(  i\right)  }$ and%
\[
S_{t}^{i}\left\{  e^{-r\left(  t+\Delta t\right)  }T_{t}^{\left(  i\right)
}-e^{-rt}T_{t}^{\left(  i\right)  }+m_{i}\Delta t\right\}  /\left[  e^{r\Delta
t}-1\right]
\]
units of cash in a risk-free bank account. \ 
\end{proposition}

\begin{proof}
The proof is included in Appendix \ref{AppendixProp3}. \ 
\end{proof}

If $\Delta t$ is not negligible compared to $\Delta S_{t}$, assuming
$\sigma=0$ and there is only one jump of $X$ between times $t$ and $t+\Delta
t$ as before, we have from (\ref{Ss})%
\begin{equation}
\Delta S_{t}=S_{t}\left[  \exp\left(  b\Delta t\right)  \left(  1+\Delta
X_{t}\right)  -1\right]  . \label{DeltaS}%
\end{equation}
Note that if $\Delta t\rightarrow0,$ $\exp\left(  b\Delta t\right)
\rightarrow1$, we have $\Delta S_{t}=S_{t}\left(  \Delta X_{t}\right)  $, as
in the case above$.$ \ Squaring both sides, we have
\[
\left(  \Delta S_{t}\right)  ^{2}=S_{t}^{2}\left\{  e^{2b\Delta t}\left(
\Delta X_{t}\right)  ^{2}+2e^{b\Delta t}\left[  e^{b\Delta t}-1\right]  \Delta
X_{t}+\left[  e^{b\Delta t}-1\right]  ^{2}\right\}  .
\]
Substituting $\Delta X_{t}$ by $\left[  \frac{\Delta S_{t}}{S_{t}}+1\right]
\exp\left(  -b\Delta t\right)  -1,$ we have%
\[
\left(  \Delta S_{t}\right)  ^{2}=2S_{t}\left[  \exp\left(  b\Delta t\right)
-1\right]  \Delta S_{t}+S_{t}^{2}\exp\left(  2b\Delta t\right)  \left(  \Delta
X_{t}\right)  ^{2}-S_{t}^{2}\left[  \exp\left(  b\Delta t\right)  -1\right]
^{2}.
\]
Similarly to (\ref{PJAD}) above,
\begin{equation}
\left(  \Delta S_{t}\right)  ^{2}=-S_{t}^{2}\left[  e^{b\Delta t}-1\right]
^{2}+2S_{t}\left[  e^{b\Delta t}-1\right]  \Delta S_{t}+S_{t}^{2}e^{2b\Delta
t}\left[  e^{-r\left(  t+\Delta t\right)  }T_{t+\Delta t}^{\left(  2\right)
}-e^{-rt}T_{t}^{\left(  2\right)  }+m_{2}\Delta t\right]  . \label{PJAD2}%
\end{equation}
We can then hedge the term $Q_{2}$ by constructing $\mathcal{P}_{t}^{\left(
2\right)  }$:

\begin{proposition}
\label{Proposition4}If $\Delta t$ is not negligible compared to $\Delta S_{t}
$, to hedge the term $Q_{2},$ we invest in $C_{2}$ units of $\mathcal{P}%
_{t}^{\left(  2\right)  }$, consisting of $S_{t}^{2}e^{2b\Delta t}e^{-r\left(
t+\Delta t\right)  }$ units of $T_{t}^{\left(  2\right)  }$ and
\begin{align*}
&  \left\{  S_{t}^{2}e^{2b\Delta t-r\left(  t+\Delta t\right)  }T_{t}^{\left(
2\right)  }-S_{t}^{2}\left[  e^{b\Delta t}-1\right]  ^{2}+2S_{t}\left[
e^{b\Delta t}-1\right]  \Delta S_{t}\right. \\
&  \ \ \ \ \ +\left.  S_{t}^{2}e^{2b\Delta t}\left[  -e^{-rt}T_{t}^{\left(
2\right)  }+m_{2}\Delta t\right]  \right\}  /\left[  e^{r\Delta t}-1\right]
\end{align*}
units of cash in a risk-free bank account.
\end{proposition}

\begin{proof}
The proof is similar to that of Proposition \ref{Proposition3}. \ 
\end{proof}

To hedge $C_{i}$ for $i>2$ if $\Delta t$ is not negligible compared to $\Delta
S_{t}$, we start from (\ref{DeltaS}),
\[
\left(  \Delta S_{t}\right)  ^{i}=S_{t}^{i}\left\{  \sum_{j=0}^{i}\binom{i}%
{j}\left(  -1\right)  ^{i-j}\exp\left(  jb\Delta t\right)  \left[  1+j\Delta
X_{t}+\sum_{k=2}^{j}\binom{j}{k}\left(  \Delta X_{t}\right)  ^{k}\right]
\right\}  .
\]
Substituting $\Delta X_{t}$ by $\left[  \frac{\Delta S_{t}}{S_{t}}+1\right]
\exp\left(  -b\Delta t\right)  -1$, we have%
\[
\left(  \Delta S_{t}\right)  ^{i}=S_{t}^{i}\sum_{j=0}^{i}\binom{i}{j}\left(
-1\right)  ^{i-j}e^{jb\Delta t}\left\{  1+j\left(  e^{-b\Delta t}-1\right)
+je^{-b\Delta t}\frac{\Delta S_{t}}{S_{t}}+\sum_{k=2}^{j}\binom{j}{k}\left(
\Delta X_{t}\right)  ^{k}\right\}  .
\]
Let
\begin{align}
c_{0}^{\left(  i,j\right)  }  &  =S_{t}^{i}\binom{i}{j}\left(  -1\right)
^{i-j}\exp\left(  jb\Delta t\right)  \left\{  1+j\left(  \exp\left(  -b\Delta
t\right)  -1\right)  \right\} \label{cij0}\\
c_{1}^{\left(  i,j\right)  }  &  =S_{t}^{i-1}\binom{i}{j}\left(  -1\right)
^{i-j}j\exp\left(  \left(  j-1\right)  b\Delta t\right) \label{cij1}\\
c_{k}^{\left(  i,j\right)  }  &  =S_{t}^{i}\binom{i}{j}\left(  -1\right)
^{i-j}\exp\left(  jb\Delta t\right)  \binom{j}{k}\text{ \ \ \ for
}k=2,3,...,j\text{ ,} \label{cijk}%
\end{align}
we have
\[
\left(  \Delta S_{t}\right)  ^{i}=\sum_{j=0}^{i}\left[  c_{1}^{\left(
i,j\right)  }\Delta S_{t}+\sum_{k=2}^{j}c_{k}^{\left(  i,j\right)  }\left(
\Delta X_{t}\right)  ^{k}+c_{0}^{\left(  i,j\right)  }\right]  .
\]
Similar to (\ref{PJAD}) above,%
\[
\left(  \Delta S_{t}\right)  ^{i}=\sum_{j=0}^{i}\left[  c_{1}^{\left(
i,j\right)  }\Delta S_{t}+\sum_{k=2}^{j}c_{k}^{\left(  i,j\right)  }\left[
\exp\left(  -r\left(  t+\Delta t\right)  \right)  T_{t+\Delta t}^{\left(
k\right)  }-\exp\left(  -rt\right)  T_{t}^{\left(  k\right)  }+m_{k}\Delta
t\right]  +c_{0}^{\left(  i,j\right)  }\right]  .
\]
Therefore, we can hedge the term $Q_{i}$ by constructing $\mathcal{P}%
_{t}^{\left(  i\right)  }$:

\begin{proposition}
\label{Proposition5}To hedge $Q_{i}$ for $i>2$ if $\Delta t$ is not negligible
compared to $\Delta S_{t}$, we invest in $C_{i}$ units of $\mathcal{P}%
_{t}^{\left(  i\right)  }$, consisting of $\sum_{j=k}^{i}c_{k}^{\left(
i,j\right)  }\exp\left(  -r\left(  t+\Delta t\right)  \right)  $ units of
$T_{t}^{\left(  k\right)  }$ for $k=2,3,...i,$ and
\begin{align*}
&  \frac{1}{\left[  \exp\left(  r\Delta t\right)  -1\right]  }\sum_{j=0}%
^{i}\left\{  \sum_{k=2}^{j}c_{k}^{\left(  i,j\right)  }\exp\left(  -r\left(
t+\Delta t\right)  \right)  T_{t}^{\left(  k\right)  }\right. \\
&  \ \ \ \ \ +\left.  c_{1}^{\left(  i,j\right)  }\Delta S_{t}+\sum_{k=2}%
^{j}c_{k}^{\left(  i,j\right)  }\left[  -\exp\left(  -rt\right)
T_{t}^{\left(  k\right)  }+m_{k}\Delta t\right]  +c_{0}^{\left(  i,j\right)
}\right\}
\end{align*}
units of cash in a risk-free bank account, where $c_{0}^{\left(  i,j\right)
}$, $c_{1}^{\left(  i,j\right)  }$ and $c_{k}^{\left(  i,j\right)  }$ are
defined in (\ref{cij0})-(\ref{cijk}).
\end{proposition}

\begin{proof}
The proof is similar to that of Proposition \ref{Proposition3}. \ 
\end{proof}

\paragraph{The general case}

In the case where there are infinite number of jumps from $t~$to $t+\Delta t$,
we need the following results on explicit formulae of CRP proved by
\cite{y06}. \ Let%
\begin{equation}
\mathcal{I}_{k}=\left\{  \left(  i_{1},i_{2},...,i_{l}\right)  |\ l\in\left\{
1,2,...,k\right\}  ,i_{q}\in\left\{  1,2,...,k\right\}  \text{ and }\sum
_{q=1}^{l}i_{q}\leq k\right\}  \label{setI}%
\end{equation}
and
\begin{equation}
\mathcal{L}_{k}=\left\{  \left(  i_{1},i_{2},...,i_{l}\right)  |l\in\left\{
1,2,...,k\right\}  ,i_{q}\in\left\{  1,2,...,k\right\}  ,i_{1}\geq i_{2}%
\geq\cdots\geq i_{l},\sum_{q=1}^{l}i_{q}=k\right\}  . \label{Lk}%
\end{equation}
The number of distinct values in a tuple $\phi_{k}=\left(  i_{1}^{(k)}%
,i_{2}^{(k)},...,i_{l}^{(k)}\right)  $ in $\mathcal{L}_{k}$ is less than or
equal to $l.$ \ When it is less than $l,$ it means some of the value(s) in the
tuple are repeated. \ Let the number of times $r\in\left\{
1,2,3,..,k\right\}  $ appears in the tuple $\phi_{k}=\left(  i_{1}^{(k)}%
,i_{2}^{(k)},...,i_{l}^{(k)}\right)  $ be $p_{r}^{\phi_{k}}.$ \ Denote the
terms which do not contain any stochastic integral in $\left(  X_{t+\Delta
t}-X_{t}\right)  ^{k}$ by $C_{\Delta t,\sigma}^{\left(  k\right)  } $.

\begin{proposition}
\label{TheoremC2}%
\begin{equation}
C_{\Delta t,\sigma}^{\left(  k\right)  }=\sum_{\phi_{k}=\left(  i_{1}%
^{(k)},i_{2}^{(k)},...,i_{l}^{(k)}\right)  \in\mathcal{L}_{k}}\frac{1}%
{l!}\left(  i_{1}^{(k)},i_{2}^{(k)},...,i_{l}^{(k)}\right)  !\left(
p_{1}^{\phi_{k}},p_{2}^{\phi_{k}},...,p_{k}^{\phi_{k}}\right)  !\left[
\prod\limits_{q\in\phi_{k}}m_{q}^{\prime}\right]  t^{l}, \label{C2_crp}%
\end{equation}
where $i_{1}^{(k)},...,i_{l}^{(k)}$ are the elements of $\phi_{k}$,
$p_{j}^{\phi_{k}}$'s are defined above and $\left(  i_{1}^{(k)},i_{2}%
^{(k)},...,i_{l}^{(k)}\right)  !$ is the multinomial coefficient: $\left(
i_{1}^{(k)},i_{2}^{(k)},...,i_{l}^{(k)}\right)  !=\frac{\left(  \sum_{j=1}%
^{l}i_{j}^{(k)}\right)  !}{i_{1}^{(k)}!i_{2}^{(k)}!\cdots i_{l}^{(k)}!}$,
$m_{q}^{\prime}=m_{q}$ for $q\neq2$ and $m_{2}^{\prime}=m_{2}+\sigma^{2}.$
\end{proposition}

Denote the coefficient of the stochastic integral $\int_{t}^{t+\Delta t}%
\int_{t}^{t_{1}-}\cdots\int_{t}^{t_{j-1}-}$\textrm{d}$Y_{t_{j}}^{\left(
i_{1}\right)  }\cdots$\textrm{d}$Y_{t_{2}}^{\left(  i_{j-1}\right)  }%
$\textrm{d}$Y_{t_{1}}^{\left(  i_{j}\right)  }$ in $\left(  X_{t+\Delta
t}-X_{t}\right)  ^{k}$ by $\Pi_{(i_{1},i_{2},...,i_{j}),\Delta t,\sigma
}^{\left(  k\right)  }. $ \ We then have the following result. \ \ \qquad
\qquad\qquad\qquad\qquad\qquad\qquad\qquad\qquad\qquad\qquad\qquad\qquad
\qquad\qquad\qquad\qquad\qquad\qquad\qquad\qquad\qquad\qquad\qquad\qquad
\qquad\qquad\qquad\qquad\qquad

\begin{proposition}
\label{Theorem2}%
\begin{equation}
\Pi_{(i_{1},i_{2},...,i_{j}),\Delta t,\sigma}^{\left(  k\right)  }=\left(
i_{1},i_{2},...,i_{j},n\right)  !C_{\Delta t,\sigma}^{\left(  n\right)
}\text{ where }n=k-\sum_{p=1}^{j}i_{p}\text{.} \label{PI}%
\end{equation}

\end{proposition}

\begin{theorem}
\label{newFormula}For any L\'{e}vy process $X$, the representation of $\left(
X_{t+\Delta t}-X_{t}\right)  ^{n}$ is given by
\[
\left(  X_{t+\Delta t}-X_{t}\right)  ^{n}=\sum_{\theta_{n}\in\mathcal{I}_{n}%
}\Pi_{\theta_{n},\Delta t,\sigma}^{\left(  n\right)  }\mathcal{S}_{\theta
_{n},\Delta t,t}^{\prime}+C_{\Delta t,\sigma}^{\left(  n\right)  },
\]
where $\mathcal{I}_{n}$ is defined in (\ref{setI}), $\Pi_{\theta_{n},\Delta
t,\sigma}^{\left(  n\right)  }$ and $C_{\Delta t,\sigma}^{\left(  n\right)  }$
are defined above and $\mathcal{S}_{\left(  i_{1},i_{2},...,i_{j}\right)
,\Delta t,t}^{\prime}$ is defined to be the integral
\[
\mathcal{S}_{\left(  i_{1},i_{2},...,i_{j}\right)  ,\Delta t,t}^{\prime}%
=\int_{t}^{t+\Delta t}\int_{t}^{t_{1}-}\cdots\int_{t}^{t_{j-1}-}%
\mathrm{d}Y_{t_{j}}^{\left(  i_{1}\right)  }\cdots\mathrm{d}Y_{t_{2}}^{\left(
i_{j-1}\right)  }\mathrm{d}Y_{t_{1}}^{\left(  i_{j}\right)  }.
\]

\end{theorem}

\ If $\Delta t$ is negligible compared to $\Delta S_{t}$, from (\ref{S}) and
Theorem \ref{newFormula},\ we have%
\begin{equation}
\left(  \Delta S_{t}\right)  ^{n}=S_{t}^{n}\left(  \Delta X_{t}\right)
^{n}=S_{t}^{n}\left(  X_{t+\Delta t}-X_{t}\right)  ^{n}=S_{t}^{n}\left[
\sum_{\theta_{n}\in\mathcal{I}_{n}}\Pi_{\theta_{n},\Delta t,\sigma}^{\left(
n\right)  }\mathcal{S}_{\theta_{n},\Delta t,t}^{\prime}+C_{\Delta t,\sigma
}^{\left(  n\right)  }\right]  . \label{DeltaSN}%
\end{equation}
\ In order to hedge $\left(  \Delta S_{t}\right)  ^{n}$, we can invest in the
\textit{power jump integral process}:%
\[
\mathcal{U}_{\left(  i_{1},i_{2},...,i_{j}\right)  ,\Delta t,t}=\exp\left(
r\Delta t\right)  \mathcal{S}_{\left(  i_{1},i_{2},...,i_{j}\right)  ,\Delta
t,t}^{\prime}.
\]
Note that since $Y^{\left(  i\right)  }$'s are martingales, $\left\{
\mathcal{S}_{\left(  i_{1},i_{2},...,i_{j}\right)  ,\Delta t,t}^{\prime}%
,t\geq0\right\}  $'s are also martingales. \ Therefore, the discounted
versions of the $\mathcal{U}_{\left(  i_{1},i_{2},...,i_{j}\right)  ,\Delta
t,t}$ are $Q$-martingales:%
\[
E_{Q}\left[  \exp\left(  -r\Delta t\right)  \mathcal{U}_{\left(  i_{1}%
,i_{2},...,i_{j}\right)  ,\Delta t,t}|\mathcal{F}_{s}\right]  =E_{Q}\left[
\mathcal{S}_{\left(  i_{1},i_{2},...,i_{j}\right)  ,\Delta t,t}^{\prime
}|\mathcal{F}_{s}\right]  =\mathcal{S}_{\left(  i_{1},i_{2},...,i_{j}\right)
,s-t,t}^{\prime},\ \ \ t\leq s\leq t+\Delta t.
\]
Hence the market allowing trade in the bond, the stock and the power jump
integral assets remains arbitrage-free. \ From (\ref{DeltaSN}), we have
$\left(  \Delta S_{t}\right)  ^{n}=S_{t}^{n}\left[  \sum_{\theta_{n}%
\in\mathcal{I}_{n}}\Pi_{\theta_{n},\Delta t,\sigma}^{\left(  n\right)  }%
\exp\left(  -r\Delta t\right)  \mathcal{U}_{\theta_{n},\Delta t,t}+C_{\Delta
t,\sigma}^{\left(  n\right)  }\right]  .$

\begin{proposition}
\label{Proposition3G}If $\Delta t$ is negligible compared to $\Delta S_{t}$,
to hedge $Q_{i},$ we invest in $C_{i}$ units of $\mathcal{P}_{t}^{\left(
i\right)  }$, consisting of $S_{t}^{i}\Pi_{\theta_{i},\Delta t,\sigma
}^{\left(  i\right)  }\exp\left(  -r\Delta t\right)  $ units of $\mathcal{U}%
_{\theta_{i},\Delta t,t}$ for $\theta_{i}\in\mathcal{I}_{i}$ and $\frac
{S_{t}^{i}C_{\Delta t,\sigma}^{\left(  i\right)  }}{\left(  \exp\left(
r\Delta t\right)  -1\right)  }$ units of cash in a risk-free bank account. \ 
\end{proposition}

\begin{remark}
In this general case, we can only derive simple hedging strategy when $\Delta
t$ is negligible. \ Note that both power jump assets introduced by
\cite{cns05} and power jump integral assets introduced here are imaginary
assets. \ In reality, we only observe a discrete series of stock price, $S$,
while there are an infinite number of jumps between any finite time interval
if the underlying L\'{e}vy process has infinite activity. \ In other words,
the values of these assets cannot be observed in the market and hence cannot
be traded. \ The moment swaps introduced by \cite{s05} depend on the increment
of the underlying stock, $\Delta S$, and can hence be observed and traded in
reality. \ We include the discussion on power jump assets for theoretical interest.
\end{remark}

Alternatively, note that in $\mathcal{S}_{\left(  i_{1},i_{2},...,i_{j}%
\right)  ,\Delta t,t}^{\prime}$, the integrand $\int_{t}^{t_{1}-}\cdots
\int_{t}^{t_{j-1}-}\mathrm{d}Y_{t_{j}}^{\left(  i_{1}\right)  }\cdots
\mathrm{d}Y_{t_{2}}^{\left(  i_{j-1}\right)  }$ is a predictable function.
\ Since we assume $\Delta t$ to be very small, we can hedge $\left(  \Delta
S_{t}\right)  ^{n}$ by investing in the power jump assets. \ Let $\phi
_{j,s}^{\left(  n\right)  }$ be the predictable function such that%
\begin{equation}
\left(  \Delta S_{t}\right)  ^{n}=S_{t}^{n}\left[  \sum_{\theta_{n}%
\in\mathcal{I}_{n}}\Pi_{\theta_{n},\Delta t,\sigma}^{\left(  n\right)
}\mathcal{S}_{\theta_{n},\Delta t,t}^{\prime}+C_{\Delta t,\sigma}^{\left(
n\right)  }\right]  =\sum_{j=1}^{n}\int_{t}^{t+\Delta t}\phi_{j,s}^{\left(
n\right)  }\mathrm{d}Y_{s}^{\left(  j\right)  }+S_{t}^{n}C_{\Delta t,\sigma
}^{\left(  n\right)  }\text{,} \label{phijsn}%
\end{equation}
where $\phi_{j,s}^{\left(  n\right)  }$'s can be calculated by rearranging the
terms in $S_{t}^{n}\sum_{\theta_{n}\in\mathcal{I}_{n}}\Pi_{\theta_{n},\Delta
t,\sigma}^{\left(  n\right)  }\mathcal{S}_{\theta_{n},\Delta t,t}^{\prime}$'s.
\ We then have%
\[
\left(  \Delta S_{t}\right)  ^{n}=\int_{t}^{t+\Delta t}\sum_{j=1}^{n}%
-e^{-2rs}T_{s}^{\left(  j\right)  }\phi_{j,s}^{\left(  n\right)  }%
\mathrm{d}e^{rs}+S_{t}^{n}C_{\Delta t,\sigma}^{\left(  n\right)  }+\sum
_{j=1}^{n}\int_{t}^{t+\Delta t}\phi_{j,s}^{\left(  n\right)  }e^{-rs}%
\mathrm{d}T_{s}^{\left(  j\right)  }.
\]
Hence, to hedge $\left(  \Delta S_{t}\right)  ^{n}$, we invest $\sum_{j=1}%
^{n}-e^{-2r\Delta t}T_{t}^{\left(  j\right)  }\phi_{j,t}^{\left(  n\right)
}+\frac{S_{t}^{n}C_{\Delta t,\sigma}^{\left(  n\right)  }}{\exp\left(  r\Delta
t\right)  -1}$ in a riskless bank account and invest $\phi_{j,t}^{\left(
n\right)  }e^{-r\Delta t}$ units of $T_{t}^{\left(  i\right)  }$ for
$j=1,2,...,n$.

\subsubsection{Extension of delta and gamma hedges\label{SectionDGLit}%
\label{SectionExt}}

So far we have discussed the hedging strategies using moment swaps and power
jump assets. \ In this section, we give a brief introduction to delta and
gamma hedging strategies and extend it to obtain perfect hedging in a L\'{e}vy
market. \ Let $\Pi$ be the value of the portfolio under consideration. \ The
delta and gamma dynamic hedging strategies are constructed using a Taylor
expansion:%
\begin{equation}
\delta\Pi=\frac{\partial\Pi}{\partial S}\delta S+\frac{\partial\Pi}{\partial
t}\delta t+\frac{1}{2}\frac{\partial^{2}\Pi}{\partial S^{2}}\delta S^{2}%
+\frac{1}{2}\frac{\partial^{2}\Pi}{\partial t^{2}}\delta t^{2}+\frac
{\partial^{2}\Pi}{\partial S\partial t}\delta S\delta t+..., \label{mt}%
\end{equation}
where $\delta\Pi$ and $\delta S$ are the changes in $\Pi$ and $S$ in a small
time interval $\delta t$. \ \cite[Chapter 14]{h03} gave detailed descriptions
of the strategies in finance. \ The \textit{delta} of a portfolio is defined
as $\frac{\partial\Pi}{\partial S}.$ \ Delta hedging eliminates the first term
on the right-hand side of (\ref{mt}). \ The \textit{gamma} of a portfolio is
defined as $\frac{\partial^{2}\Pi}{\partial S^{2}}.$ \ Gamma hedging
eliminates the third term on the right-hand side of (\ref{mt}).

Below we extend the gamma hedge in order to obtain a perfect hedging strategy
in a L\'{e}vy market. \ Note that equation (\ref{mt}) is a multivariate Taylor
expansion and it is assumed that all the cross derivative terms are
negligible. \ In equation (\ref{f2}), we applied Taylor expansions twice to
avoid the cross derivative terms, since the value of $\Delta t$ is
{deterministic} and known at time $t$. \ Hence, for fixed $n$, the
approximation by:
\begin{equation}
F\left(  t+\Delta t,S_{t}+\Delta S_{t}\right)  -F\left(  t,S_{t}\right)
=\sum_{i=1}^{\infty}\frac{D_{1}^{i}F\left(  t,S_{t}\right)  }{i!}\left(
\Delta t\right)  ^{i}+\sum_{i=1}^{n}\frac{D_{2}^{i}F\left(  t+\Delta
t,S_{t}\right)  }{i!}\left(  \Delta S_{t}\right)  ^{i} \label{f4}%
\end{equation}
is more accurate than%
\[
F\left(  t+\Delta t,S_{t}+\Delta S_{t}\right)  -F\left(  t,S_{t}\right)
=\sum_{i=1}^{\infty}\frac{D_{1}^{i}F\left(  t,S_{t}\right)  }{i!}\left(
\Delta t\right)  ^{i}+\sum_{i=1}^{n}\frac{D_{2}^{i}F\left(  t,S_{t}\right)
}{i!}\left(  \Delta S_{t}\right)  ^{i}.
\]
Moreover, in the literature, $\Delta t$ and $\Delta S$ are assumed to be very
small (such that the cross terms and higher terms are negligible). \ We
provide the flexibility of specifying the values of $\Delta t$ and $\Delta
S_{t}$ such that static hedging is possible in some cases.

It is natural to extend the delta and gamma hedging strategies to the $n$-th
derivative of the portfolio with respect to the underlying asset using the
approximation of equation (\ref{f4}). \ Let $F$ be the value of our portfolio
to be hedged and there are $n-1$ traded financial derivatives, $F_{i},$
$i=2,...,n$, which are linearly independent of each other. \ Suppose we add
$w_{i}~$number of $F_{i}$ into our portfolio, $i=2,...,n$ and add $w_{1}$
number of the underlying asset, which is denoted by $F_{1}$. \ We assume that
$D_{2}^{j}F_{i}\left(  t+\Delta t,S_{t}\right)  $ are nonzero for $j=i$ and
can be zero, or not, for $j=1,2,...,i-1,i+1,...,n.$ \ In general, to make the
portfolio $k$-th moment neutral for $k=1,...,n$, we need $D_{2}^{k}F\left(
t+\Delta t,S_{t}\right)  +\sum_{i=1}^{n}w_{i}D_{2}^{k}F_{i}\left(  t+\Delta
t,S_{t}\right)  =0$ for $k=1,2,...,n$. \ Therefore, we have $n$ equations for
$n$ unknown, $w_{i}$'s. \ Note that whether the system of equations is
solvable depends on the values of $D_{2}^{k}F_{i}\left(  t+\Delta
t,S_{t}\right)  $, $i,k=1,2,...,n.$ \ Therefore, the traded financial
derivatives have to be chosen such that the system of equations are solvable.

\section{Minimal variance portfolios in a L\'{e}vy market\label{SectionMVP}}

In Section \ref{SectionHS}, we gave the perfect hedging portfolios, given that
the moment swaps, power jump assets and certain financial derivatives that
depend on the same underlying asset, are available in the market. \ In this
section, we demonstrate how to use the minimal variance portfolios derived by
\cite{bdlop03} to hedge the higher order terms in the Taylor expansion,
investing only in a risk-free bank account, the underlying asset and, if
possible, variance swaps.

\cite{bdlop03} derived the minimal variance hedging portfolio of a contingent
claim in a market such that the stock prices are independent L\'{e}vy
martingales in terms of Malliavin derivatives. \ We demonstrate how to use
their results to hedge the terms $Q_{i}$. \ Following \cite{bdlop03}, to
derive the minimal variance portfolio, we need to confine ourselves to the
case of L\'{e}vy processes, $\eta=\left\{  \eta\left(  t\right)  ,0\leq t\leq
T\right\}  $, which are martingales on the filtered probability space under
consideration. \ That is, $E\left[  \eta\left(  t\right)  \right]  =0$ and
$E\left[  \eta^{2}\left(  t\right)  \right]  =\left(  \sigma^{2}%
+\int_{\mathbb{R}}x^{2}\nu\left(  \mathrm{d}x\right)  \right)  t$.
\ \cite{bdlop03} called such processes \textit{L\'{e}vy martingales of the
second order}. \ From \cite[equation (2.1)]{bdlop03}, $\eta\left(  t\right)  $
has the following representation formula:%
\begin{equation}
\eta\left(  t\right)  =\sigma W\left(  t\right)  +\int_{0}^{t}\int
_{\mathbb{R}}x\tilde{N}\left(  \mathrm{d}s,\mathrm{d}x\right)
,\ \ \ \text{for }0\leq t\leq T, \label{decomp}%
\end{equation}
where $\sigma\in\mathbb{R}^{+}$, $W\left(  t\right)  $ is the standard
Brownian motion and $\tilde{N}\left(  \mathrm{d}t,\mathrm{d}x\right)  $ is
defined in (\ref{cprm}).

Based on the methodology developed by \cite{bdlop03}, we modify their results
to express the minimal variance portfolio for independent securities without
referring to Malliavin calculus. \ \cite{bdlop03} assumed the underlying asset
is directly represented by the L\'{e}vy martingale, that is, $S_{t}%
=\eta\left(  t\right)  $. \ We find it more natural to employ an exponential
model and allow a drift term in the model of the underlying asset since the
mean of $\eta\left(  t\right)  $ is zero. \ By extending (\ref{S}), we suppose
there are $k$ independent securities prices $S_{1},...,S_{k} $, modeled as
follows:%
\begin{equation}
\mathrm{d}S_{j}\left(  t\right)  =b_{j}S_{j}\left(  t_{-}\right)
\mathrm{d}t+S_{j}\left(  t_{-}\right)  \mathrm{d}\eta_{j}\left(  t\right)
,\text{ \ \ }j=1,...,k, \label{S2}%
\end{equation}
where $b_{j}\in\mathbb{R}$. \ Let $L_{2}\left(  \Omega\right)  =L_{2}\left(
\Omega,\mathcal{F},P\right)  $ and $\xi\in L^{2}\left(  \Omega\right)  $ be a
random variable to be hedged. \ Let $\mathcal{A}$ be the set of all admissible
portfolios. \ The minimal variance portfolio is an admissible portfolio,
$\varphi\in\mathcal{A}$ such that%
\begin{equation}
E\left[  \left(  \xi-E\left[  \xi\right]  -\sum_{j=1}^{k}\int_{0}^{T}%
\varphi_{j}\left(  s\right)  \mathrm{d}S_{j}\left(  s\right)  \right)
^{2}\right]  =\inf_{\psi\in\mathcal{A}}E\left[  \left(  \xi-E\left[
\xi\right]  -\sum_{j=1}^{k}\int_{0}^{T}\psi_{j}\left(  s\right)
\mathrm{d}S_{j}\left(  s\right)  \right)  ^{2}\right]  . \label{mvh}%
\end{equation}
This is known as the \textit{minimal variance hedging} for incomplete markets.
\ Define a measure of the length of $\xi$ by $\left\Vert \xi\right\Vert
=\left(  \int_{\Omega}\left\vert \xi\left(  \omega\right)  \right\vert
^{2}P\left(  \mathrm{d}\omega\right)  \right)  ^{1/2}=\left(  E\left[
\left\vert \xi\right\vert ^{2}\right]  \right)  ^{1/2}.$ \ Following
\cite[Definition 3.10 (a)]{bdlop03}, let $\mathbb{D}_{1,2}$ be the set of all
$\xi\in L_{2}\left(  \Omega\right)  $ such that the chaos expansion defined in
(\ref{gprm}) satisfies the condition%
\[
\left\Vert \xi\right\Vert _{\mathbb{D}_{1,2}}^{2}=E\left[  \xi^{2}\right]
+\sum_{n=1}^{\infty}\sum_{j_{1},...,j_{n}=1,2}\int_{U_{j_{n}}}\left\Vert
g_{n}^{\left(  j_{1},...,j_{n}\right)  }\left(  \cdot,u_{n}^{\left(
j_{n}\right)  }\right)  \right\Vert _{L_{2}\left(  G_{n-1}\right)  }%
^{2}\mathrm{d}\left\langle Q_{j_{n}}\right\rangle \left(  u_{n}^{\left(
j_{n}\right)  }\right)  <\infty,
\]
where $G_{n}$ is defined in (\ref{gn}). \ The chaotic representation derived
by \cite{bdlop03} implies that every $\xi$ satisfying some moment conditions
can be expressed in the form%
\begin{equation}
\xi=E\left[  \xi\right]  +\sum_{j=1}^{k}\int_{0}^{T}f_{1}\left(
\xi;s,j\right)  \mathrm{d}W_{j}\left(  s\right)  +\sum_{j=1}^{k}\int_{0}%
^{T}\int_{\mathbb{R}}f_{2}\left(  \xi;s,x,j\right)  \tilde{N}_{j}\left(
\mathrm{d}s,\mathrm{d}x\right)  , \label{f1f2}%
\end{equation}
where $f_{1}\left(  \xi;s,j\right)  $ and $f_{2}\left(  \xi;s,x,j\right)  $
are predictable functions. \ \cite{y06} derived the computationally explicit
representation formula for $f_{1}\left(  \xi;s,j\right)  $ and $f_{2}\left(
\xi;s,x,j\right)  $ when $\xi$ is the power of increments of a L\'{e}vy
process, see Theorem \ref{newFormula}. \ The minimal variance portfolio
consisting of independent securities driven by (\ref{S2}), can be obtained by
modifying Theorem 4.1 in \cite{bdlop03}:

\begin{proposition}
\label{PropositionMVP}For any $\xi\in\mathbb{D}_{1,2}$, the minimal variance
portfolio $\varphi=\left(  \varphi_{1},...,\varphi_{k}\right)  $ in
(\ref{mvh}),
\[
\hat{\xi}=E\left[  \xi\right]  +\sum_{j=1}^{k}\int_{0}^{T}\varphi_{j}\left(
s\right)  \mathrm{d}S_{j}\left(  s\right)  ,
\]
admits the following representation:
\[
\varphi_{j}\left(  s\right)  =\frac{f_{1}\left(  \xi;s,j\right)  \sigma
_{j}+\int_{\mathbb{R}}xf_{2}\left(  \xi;s,x,j\right)  \nu_{j}\left(
\mathrm{d}x\right)  }{\left\{  \sigma_{j}^{2}+\int_{\mathbb{R}}x^{2}\nu
_{j}\left(  \mathrm{d}x\right)  \right\}  S_{j}\left(  s\right)  },
\]
where $f_{1}\left(  \xi;s,j\right)  $ and $f_{2}\left(  \xi;s,x,j\right)  $
are predictable functions defined in (\ref{f1f2}).
\end{proposition}

\begin{proof}
The proof is included in Appendix \ref{AppendixProofProp1}.
\end{proof}

Although variance swaps are traded in OTC markets, there might be times that
the appropriate variance swaps needed are not available. \ Hence, we firstly
discuss how to use a minimal variance portfolio to hedge $\sum_{i=2}^{q}Q_{i}
$ using only a risk-free bank account and the underlying stock. \ As in
Section \ref{SectionHPJ}, we consider the simplified case where there is at
most one jump of $X$ between $t$ and $t+\Delta t$, and the general case where
there can be infinite number of jumps. \ 

\subsection{The simplified case}

If $\Delta t$ is negligible compared to $\Delta S_{t},$ from (\ref{PJAD}),
\begin{equation}
\sum_{i=2}^{q}Q_{i}=\sum_{i=2}^{q}C_{i}S_{t}^{i}\left[  \int_{t}^{t+\Delta
t}\mathrm{d}Y_{s}^{\left(  i\right)  }+m_{i}\Delta t\right]  . \label{PJADi}%
\end{equation}

\begin{proposition}
\label{PropositionMVPwithout}If $\Delta t$ is negligible compared to $\Delta
S_{t}$, the minimal variance portfolio to hedge $\sum_{i=2}^{q}Q_{i}$ using
only a risk-free bank account and the underlying asset is to\newline1) invest
$\sum_{i=2}^{q}\frac{C_{i}}{\left(  \exp\left(  r\Delta t\right)  -1\right)
}S_{t}^{i}m_{i}\Delta t$ in a risk-free bank account, and\newline2) buy
$\frac{1}{\left[  \sigma^{2}+m_{2}\right]  }\sum_{i=2}^{q}C_{i}S_{t}%
^{i-1}m_{i+1}$ units of the underlying stock, $S_{t}, $ where $m_{i}$ is
defined in (\ref{mean}).
\end{proposition}

\begin{proof}
The proof is included in Appendix \ref{AppendixMVPwithout}.
\end{proof}

In the following, we discuss how to hedge the terms $\sum_{i=3}^{q}Q_{i}$
using a risk-free bank account, the underlying stock and variance swaps. \ If
$\Delta t$ is negligible compared to $\Delta S_{t},$ from (\ref{PJAD}),
\begin{equation}
\sum_{i=3}^{q}Q_{i}=\sum_{i=3}^{q}C_{i}S_{t}^{i}\left[  \int_{t}^{t+\Delta
t}\mathrm{d}Y_{s}^{\left(  i\right)  }+m_{i}\Delta t\right]  . \label{C1}%
\end{equation}
Therefore, we have the following hedging portfolio.

\begin{proposition}
\label{PropositionMVPwith}If $\Delta t$ is negligible compared to $\Delta
S_{t},$ the minimal variance portfolio to hedge $\sum_{i=3}^{q}Q_{i}$ by
investing in a risk-free bank account, the underlying asset and variance swaps
is given by:\newline1) buy $\phi\Delta s\left(  n-2\right)  S_{t}^{2}$ units
of the variance swap at time $t$ with sampling points $\left\{  ...,s_{n-1}%
=t,s_{n}=t+\Delta t\right\}  \,$, maturity $t+\Delta t$ and strike
$\sigma_{\text{strike}}^{2}$, where%
\[
\phi=\frac{\sum_{i=3}^{q}C_{i}S_{t}^{i-2}\int_{\mathbb{R}}x^{i}\nu\left(
\mathrm{d}x\right)  }{\int_{\mathbb{R}}x^{2}\nu\left(  \mathrm{d}x\right)
}=\frac{\sum_{i=3}^{q}C_{i}S_{t}^{i-2}m_{i}}{m_{2}},
\]
$m_{i}$ are defined in (\ref{mean}) and $P_{V}$ is the price of one unit of
the variance swap.\newline2) invest nothing in the underlying asset, $S_{t}%
$,\newline3) invest
\[
\frac{1}{e^{r\Delta t}-1}\left\{  \sum_{i=3}^{q}C_{i}S_{t}^{i}m_{i}\Delta
t+\phi S_{t}^{2}\left\{  \Delta s\left(  n-2\right)  \left[  \sigma
_{\text{strike}}^{2}-\overline{S}_{n,2}\right]  +P_{V}\Delta s\left(
n-2\right)  -m_{2}\Delta t\right\}  \right\}
\]
in a risk-free bank account, where $\overline{S}_{n,2}$ is defined in
(\ref{S2bar}).
\end{proposition}

\begin{proof}
The proof is similar to those of Propositions \ref{PropositionMVP} and
\ref{PropositionMVPwithout}.
\end{proof}

\subsection{The general case}

If $\Delta t$ is negligible compared to $\Delta S_{t},$ from (\ref{DeltaSN}),%
\[
\left(  \Delta S_{t}\right)  ^{n}=S_{t}^{n}\sum_{\theta_{n}\in\mathcal{I}_{n}%
}\Pi_{\theta_{n},\Delta t,\sigma}^{\left(  n\right)  }\mathcal{S}_{\theta
_{n},\Delta t,t}^{\prime}+S_{t}^{n}C_{\Delta t,\sigma}^{\left(  n\right)  },
\]
where the expression can be calculated explicitly using Theorem
\ref{newFormula}. \ Let%
\[
\sum_{i=2}^{q}Q_{i}=\sum_{j=1}^{q}C_{i}\int_{t}^{t+\Delta t}\phi
_{j,s}^{\left(  q\right)  }\mathrm{d}Y_{s}^{\left(  j\right)  }+\sum_{i=2}%
^{q}C_{i}S_{t}^{i}C_{\Delta t,\sigma}^{\left(  i\right)  }\text{,}%
\]
where $\phi_{j,s}^{\left(  q\right)  }$ is defined in (\ref{phijsn}). \ 

\begin{proposition}
\label{PropositionMVPwithoutG}If $\Delta t$ is negligible compared to $\Delta
S_{t}$, the minimal variance portfolio to hedge $\sum_{i=2}^{q}Q_{i}$ using
only a risk-free bank account and the underlying asset is to\newline1) invest
$\sum_{i=2}^{q}\frac{C_{i}}{\exp\left(  r\Delta t\right)  -1}S_{t}%
^{i}C_{\Delta t,\sigma}^{\left(  i\right)  }$ in a risk-free bank account,
and\newline2) buy $\frac{1}{\left[  \sigma^{2}+m_{2}\right]  }\sum_{j=1}%
^{q}C_{i}\phi_{j,s}^{\left(  q\right)  }S_{t}^{-1}m_{i+1}$ units of the
underlying stock, $S_{t},$ where $m_{i}$ is defined in (\ref{mean}).
\end{proposition}

\begin{proof}
The proof is similar to that of Proposition \ref{PropositionMVPwithout}.
\end{proof}

In the following, we discuss how to hedge the terms $\sum_{i=3}^{q}Q_{i}$
using a risk-free bank account, the underlying stock and variance swaps. \ 

\begin{proposition}
\label{PropositionMVPwithG}If $\Delta t$ is negligible compared to $\Delta
S_{t},$ the minimal variance portfolio to hedge $\sum_{i=3}^{q}Q_{i}$ by
investing in a risk-free bank account, the underlying asset and variance swaps
is given by:\newline1) buy $\phi\Delta s\left(  n-2\right)  S_{t}^{2}$ units
of the variance swap at time $t$ with sampling points $\left\{  ...,s_{n-1}%
=t,s_{n}=t+\Delta t\right\}  \,$, maturity $t+\Delta t$ and strike
$\sigma_{\text{strike}}^{2}$, where%
\[
\phi=\frac{\sum_{i=1}^{q}C_{i}\phi_{j,s}^{\left(  q\right)  }S_{t}^{-2}%
\int_{\mathbb{R}}x^{i}\nu\left(  \mathrm{d}x\right)  }{\int_{\mathbb{R}}%
x^{2}\nu\left(  \mathrm{d}x\right)  }=\frac{\sum_{i=1}^{q}C_{i}\phi
_{j,s}^{\left(  q\right)  }S_{t}^{-2}m_{i}}{m_{2}},
\]
$m_{i}$ are defined in (\ref{mean}) and $P_{V}$ is the price of one unit of
the variance swap.\newline2) invest nothing in the underlying asset, $S_{t}%
$,\newline3) invest
\[
\frac{1}{e^{r\Delta t}-1}\left\{  \sum_{i=2}^{q}C_{i}S_{t}^{i}C_{\Delta
t,\sigma}^{\left(  i\right)  }+\phi S_{t}^{2}\left\{  \Delta s\left(
n-2\right)  \left[  \sigma_{\text{strike}}^{2}-\overline{S}_{n,2}\right]
+P_{V}\Delta s\left(  n-2\right)  -m_{2}\Delta t\right\}  \right\}
\]
in a risk-free bank account, where $\overline{S}_{n,2}$ is defined in
(\ref{S2bar}).
\end{proposition}

\begin{proof}
The proof is similar to that of Proposition \ref{PropositionMVPwithout}.
\end{proof}

\section{Simulation algorithm\label{SectionDer}\label{SectionImple}}

In this section, we discuss the approximation of the derivatives, $D_{2}%
^{i}F\left(  t+\Delta t,S_{t}\right)  $, and computational implementation of
the hedging strategies. \ Assuming that the terms $\sum_{i=2}^{\infty}%
\frac{D_{1}^{i}F\left(  t,S_{t}\right)  }{i!}\left(  \Delta t\right)  ^{i}$ do
not contribute to the approximation significantly and can be ignored (which is
found to be true in our simulation study), we have%
\[
F\left(  t+\Delta t,S_{t}+\Delta S_{t}\right)  -F\left(  t,S_{t}\right)
=D_{1}^{1}F\left(  t,S_{t}\right)  \Delta t+\sum_{i=1}^{q}\frac{D_{2}%
^{i}F\left(  t+\Delta t,S_{t}\right)  }{i!}\left(  \Delta S_{t}\right)  ^{i},
\]
which is true as long as $D_{1}^{1}F\left(  t,S_{t}\right)  $ and $D_{2}%
^{i}F\left(  t+\Delta t,S_{t}\right)  $ exist for $i=1,2,3,...$. \ Note that
the assumption $\sum_{i=2}^{\infty}\frac{D_{1}^{i}F\left(  t,S_{t}\right)
}{i!}\left(  \Delta t\right)  ^{i}$ $\thickapprox0$ is only for simplicity
here since we are more interested in finding ways to hedge $\sum_{i=1}%
^{q}\frac{D_{2}^{i}F\left(  t+\Delta t,S_{t}\right)  }{i!}\left(  \Delta
S_{t}\right)  ^{i}$. \ The deterministic terms $\sum_{i=2}^{\infty}\frac
{D_{1}^{i}F\left(  t,S_{t}\right)  }{i!}\left(  \Delta t\right)  ^{i}$ can be
hedged by investing in a risk-free bank account, as in equation (\ref{P}).
\ Since the pricing formulae for options with underlying driven by L\'{e}vy
processes are in general not analytic, we need to approximate the derivatives
of the pricing formulae, $D_{2}^{i}F\left(  t+\Delta t,S_{t}\right)  $, for
$i=1,2,3,....$ \ We employ the Taylor's series based central difference
approximation of arbitrary $p$-th degree derivatives introduced by
\cite[Section 1]{ko03}, which is quoted in Appendix \ref{centralDiff}.

In the following, we discuss how to calculate the derivatives of the option
prices. \ We note that the most time consuming step in the approximation
procedures is the calculation of $\sum_{i}\frac{1}{X\left(  i\right)  ^{2}}$
in finding $d_{k}^{\left(  p\right)  }$ in equation (\ref{dp2}) in the central
difference approximation of derivatives. \ It is because the vector $X$
contains the product of all the possible combinations of length $c$ in $Y$,
where $Y$ contains all integers from $1$ to $N$ except $\left\vert
k\right\vert .$ \ For example, if we want to approximate the 31st derivative
and set $N$ $=33$ (the accuracy of the approximation increases with the value
of $N$)$,$ $c=15$ and $k=1$, the number of values in $Y$ is 32 and the number
of possible combinations of length $c$ in $Y$ is $C_{15}^{32}=\frac
{32!}{15!\left(  32-15\right)  !}=565,722,720$, which takes quite a while to
calculate. \ Nevertheless, this calculation is the same for all functions
$f\left(  t\right)  $. \ Therefore, we can build up a look-up table to store
values of $C_{N,k}\sum_{i}\frac{1}{X\left(  i\right)  ^{2}}$ for different
$N$, $c$ and $k$ and use it for all options. \ Although the calculation for
large $N$ can take a very long time, we only need to do this once.%

\begin{center}%
%

\begin{tabular}
[c]{|ll|}\hline
\multicolumn{2}{|l|}{Algorithm}\\\hline
1. & Construct the look-up table of $C_{N,k}\sum_{i}\frac{1}{X\left(
i\right)  ^{2}}$ defined in equation (\ref{dp2})$.$\\
2. & Calculate sample paths of $S$ with different values of the current
stock\\
& price, $S_{t}$.\\
3. & Use Monte Carlo simulation to calculate the option prices with respect\\
& to different values of the current stock price.\ \\
4. & Calculate the derivatives with respect to the underlying, $D_{2}%
^{i}F\left(  t+\Delta t,S_{t}\right)  $.\\
5. & Calculate the first derivative with respect to time, $D_{1}^{1}F\left(
t,S_{t}\right)  .$\\\hline
\end{tabular}
%

\end{center}%

Table \ref{SectionImple}.1: The simulation algorithm to calculate the
derivatives in Taylor expansions.

\bigskip

\textbf{Step 1\qquad}For a fixed $N$, construct the look-up table of
$C_{N,k}\sum_{i}\frac{1}{X\left(  i\right)  ^{2}}$, where $k=0$,$1$,$2$%
,$...$,$N $ and $c=3,4,...,c_{\max}$, where $c_{\max}$ $=N-1$ (since $2N>p$
and $c$ is the largest integer less than or equal to $\frac{p-1}{2}$)$.$
\ Therefore, the maximum derivative obtainable is $\left(  2N-1\right)  $-th.

Note that we should loop through $c$ and then $k$. \ For each value of $c$, we
use a vector to save the intermediate values of $\sum_{i}\frac{1}{X\left(
i\right)  ^{2}}$ for each $k.$ \ Therefore, we only need to calculate the
combination of choosing $c$ from $Y$ once for each $c$.

\textbf{Step 2}\qquad Calculate sample paths of $S$ with different values of
the current stock price, $S_{t}$. \ 

\textbf{Step 3}\qquad Use Monte Carlo simulation to calculate the option
prices with respect to different values of the current stock price, using the
sample paths of $S$ generated in Step 2.

\textbf{Step 4}\qquad Using the finite different method given in Appendix
\ref{centralDiff}, calculate the derivatives with respect to the underlying,
$D_{2}^{i}F\left(  t+\Delta t,S_{t}\right)  $, using the look-up table
produced in Step 1.

\textbf{Step 5}\qquad Similar to Step 4, calculate the first derivative with
respect to time, $D_{1}^{1}F\left(  t,S_{t}\right)  .$

After calculating the derivatives, we show the performance of the proposed
hedging strategies in the next section.

\section{Performance of the hedging strategies\label{SectionPerform}%
\label{SectionResult}}

In this section, we investigate the performance of the hedging strategies
given in Section \ref{SectionHS} on European options and barrier options. \ We
also give an example of static hedging of an one year European option on real
life data. \ We truncate the infinite sum in (\ref{f1}) and calculate
$\sum_{i=1}^{p}\frac{D_{2}^{i}F\left(  t+\Delta t,S_{t}\right)  }{i!}\left(
\Delta S_{t}\right)  ^{i}+D_{1}^{1}F\left(  t,S_{t}\right)  \Delta t$ for some
fixed $p$. \ By comparing the values on the L.H.S. and R.H.S. of (\ref{f1}),
it may be noted that for some $q\in\mathbb{N}$, the terms $\frac{D_{2}%
^{i}F\left(  t+\Delta t,S_{t}\right)  }{i!}\left(  \Delta S_{t}\right)
^{i}\simeq0$ for $i>q$. \ This approximation is very useful, since in practice
it is ideal to hedge by investing in as few kinds of products as possible, due
to cost of transaction and administration. \ By fixing a tolerance level,
$\alpha_{\text{tol}}$, we can find the smallest value of $p $ such that%
\begin{equation}
\left\vert \left[  F\left(  t+\Delta t,S_{t}+\Delta S_{t}\right)  -F\left(
t,S_{t}\right)  \right]  -\left[  D_{1}^{1}F\left(  t,S_{t}\right)  \Delta
t+\sum_{i=1}^{p}\frac{D_{2}^{i}F\left(  t+\Delta t,S_{t}\right)  }{i!}\left(
\Delta S_{t}\right)  ^{i}\right]  \right\vert \leq\alpha_{\text{tol}}
\label{alphaError}%
\end{equation}
and we call it $q$. \ For a given tolerance level, $\alpha_{\text{tol}}$, the
following approximation is then assumed satisfactory:%
\begin{equation}
F\left(  t+\Delta t,S_{t}+\Delta S_{t}\right)  -F\left(  t,S_{t}\right)
=D_{1}^{1}F\left(  t,S_{t}\right)  \Delta t+\sum_{i=1}^{q}\frac{D_{2}%
^{i}F\left(  t+\Delta t,S_{t}\right)  }{i!}\left(  \Delta S_{t}\right)  ^{i}.
\label{q2}%
\end{equation}
Thus the magnitude of $\alpha_{\text{tol}}$ determines the number of terms
required for a Taylor expansion to obtain a satisfactory approximation. \ In
option hedging, we want the number of terms to be as small as possible since
we have to invest in an additional financial derivative to hedge each term.
\ In practice as we noted before, transaction costs, bid-ask spreads and the
cost of administration make the trades of a large number of different
financial derivatives not preferable. \ Therefore, there is a trade-off
between the accuracy of the hedging and the additional costs involved.\ %

{\includegraphics[
height=2.0833in,
width=5.585in
]%
{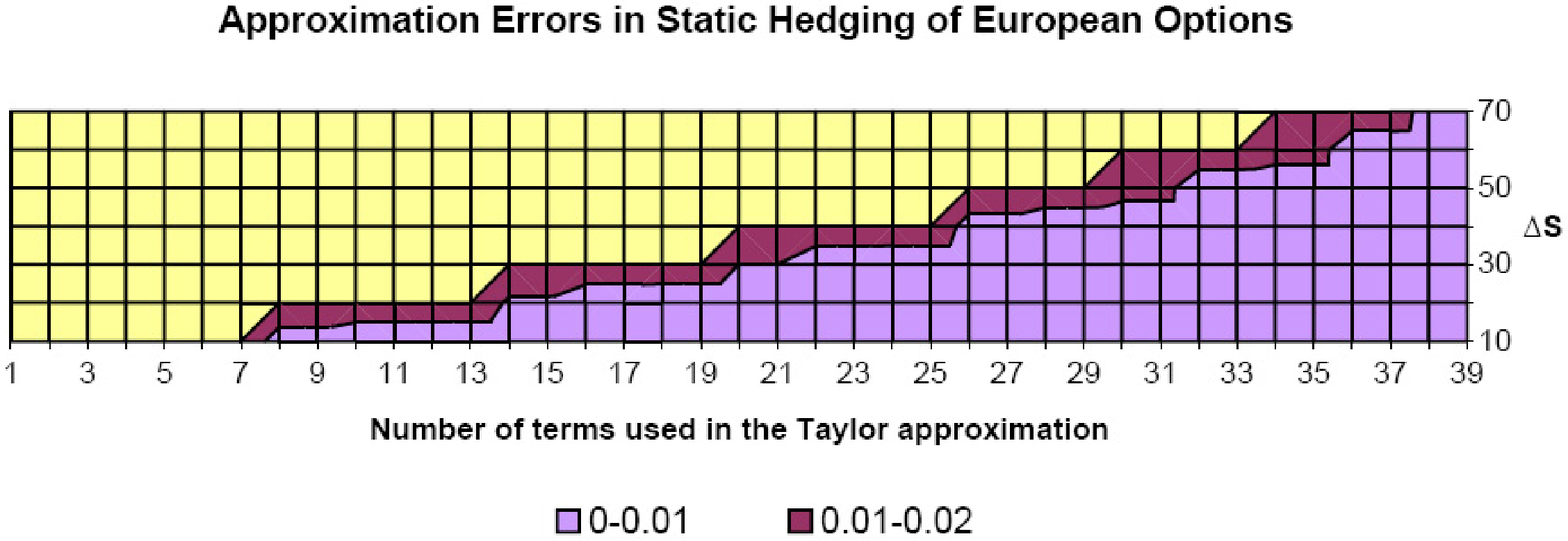}%
}%
%

{\includegraphics[
height=2.034in,
width=5.5357in
]%
{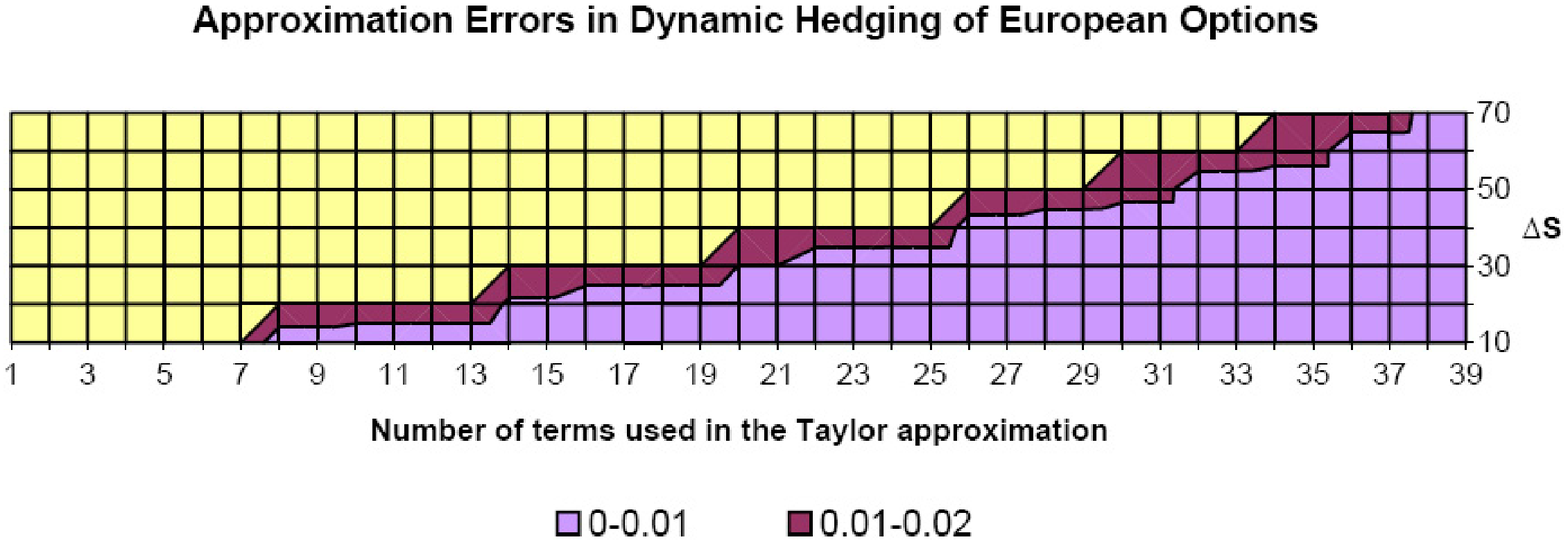}%
}%
%

{\includegraphics[
height=2.0012in,
width=5.4777in
]%
{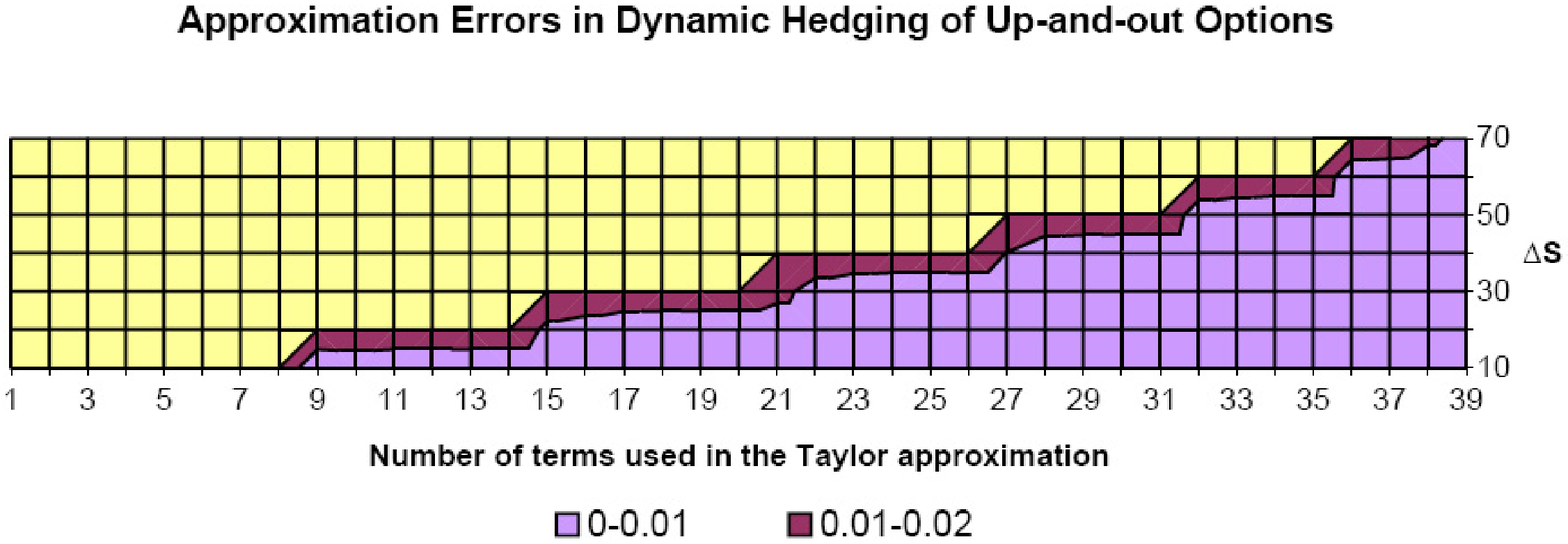}%
}%

Figure \ref{SectionResult}.1: The approximation errors in static hedging of
European options, dynamic hedging of European options and dynamic hedging of
up-and-out options. \ The $x\,$-axis gives the value of $q$ and the $y$-axis
gives $\Delta S$. \ The area of the graph is coloured in light purple when the
approximation error $\leq0.01$ and in deep purple when the approximation error
is between $0.01$ and $0.02$.

\bigskip

In the following, we give the performance of the static and dynamic hedging
strategies on European, up-and-out, up-and-in, down-and-out and down-and-in
options. \ We investigate how many terms in the Taylor expansions are needed
to obtain a satisfactory approximation, that is, we determine the value of $q
$ for a given $\alpha_{\text{tol}}$, defined in (\ref{alphaError}). \ In our
simulations, we set $\alpha_{\text{tol}}=0.01.$ \ It is because in practice,
we are hedging the prices of the options, the lowest price change is 0.01.
\ We assume the current stock price, $S_{0}$, is 5000 and the strike price of
the options, $K$, are 5000. \ Note that our strategies work for all values of
$K$. \ We consider the cases where the change in the price of the stock price
$\Delta S_{t}$ is equal to 10, 20, ..., 70. \ For static hedging, we assume
$\Delta t=1$, and the options are expiring in 1 year as well, that is, $T=1$.
\ For dynamic hedging, we set $\Delta t=9.5129\times10^{-6}$, approximately 5
minutes, and $T=1.1416\times10^{-4}$, approximately 1 hour.%

{\includegraphics[
height=1.9493in,
width=5.4855in
]%
{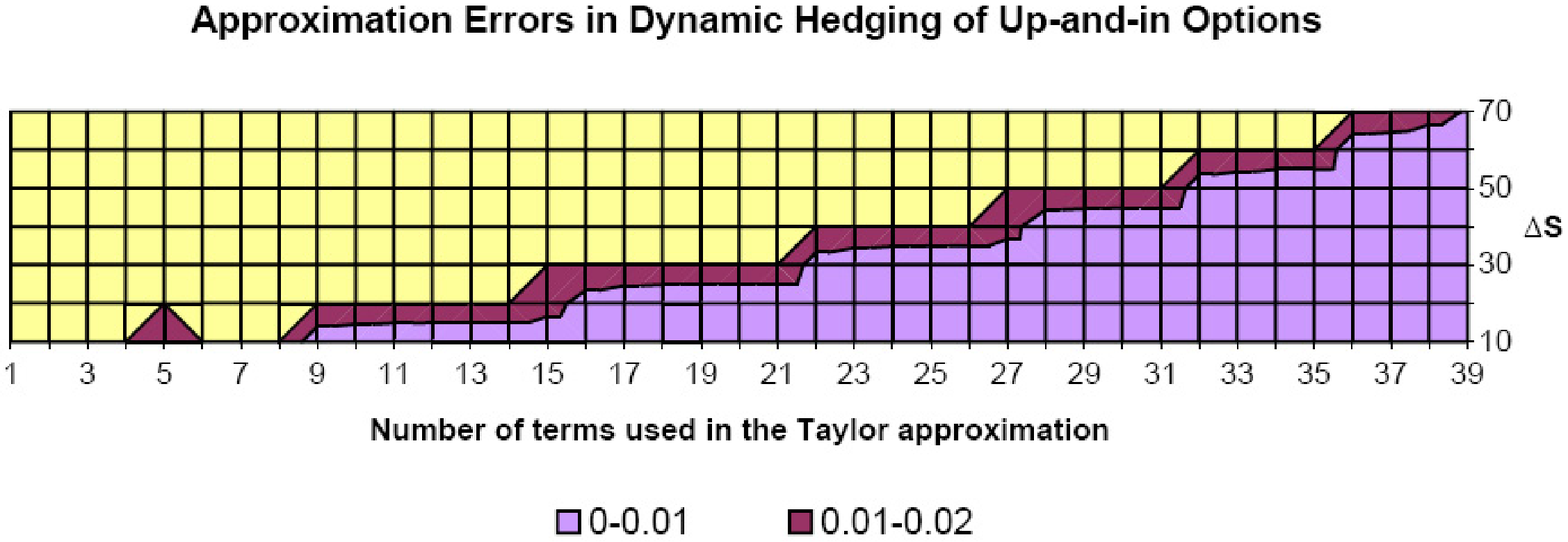}%
}%
%

{\includegraphics[
height=1.9493in,
width=5.5434in
]%
{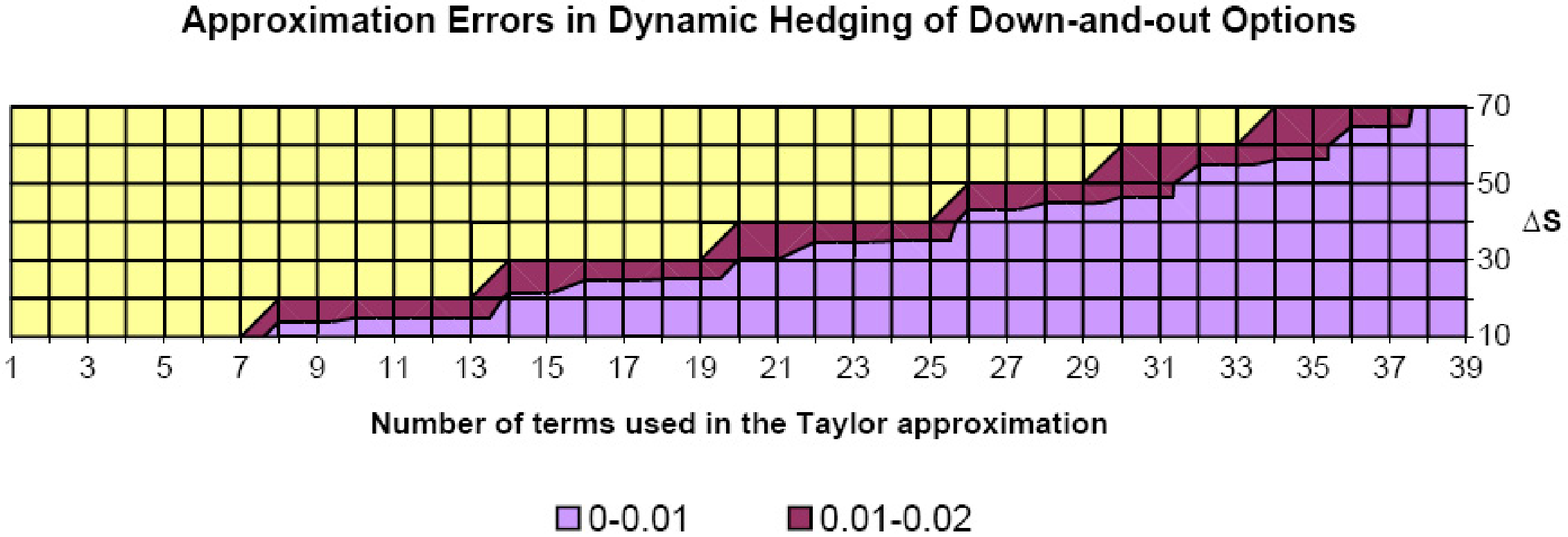}%
}%
%

{\includegraphics[
height=1.9925in,
width=5.5201in
]%
{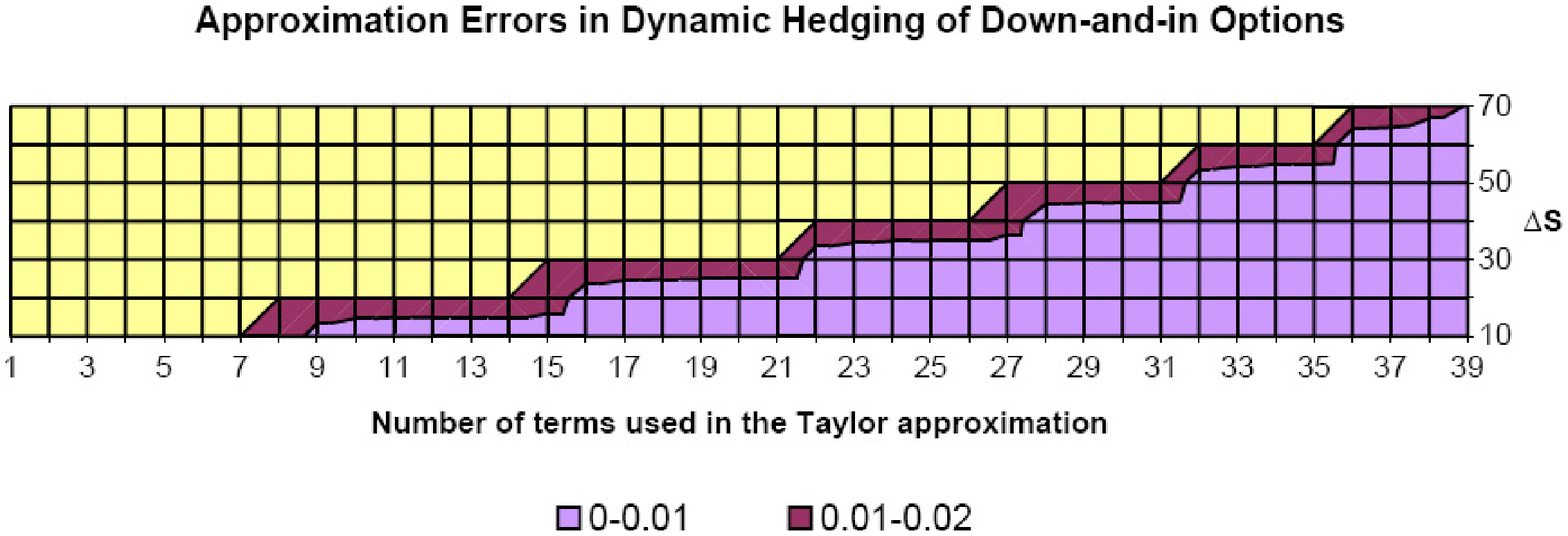}%
}%

Figure \ref{SectionResult}.2: The approximation error in dynamic hedging of
up-and-in options, down-and-out options and down-and-in options. \ The
$x\,$-axis gives the value of $q$ and the $y$-axis gives $\Delta S$.

\bigskip

The performance of static and dynamic hedging of European options is given in
Figure \ref{SectionResult}.1. \ We can see that the values of $q$ required are
the same in the cases of static and dynamic hedging. \ The value of $q$, that
is, the number of terms required in the Taylor approximation, such that the
error $\leq\alpha_{\text{tol}}$ increases gradually as the value of $\Delta
S_{t}$ increases. \ This verifies the discussion given in the beginning of
this section, that is, for a given tolerance level, the number of terms
required in the Taylor expansions is finite. \ The values of $q$ for different
values of $\Delta S_{t}$ is also given in Table \ref{SectionResult}.1.\ 

The performance of dynamically hedging of {up-and-out} options is given in
Figure \ref{SectionResult}.1. \ We assume the barrier is given by $H=5050.$
\ The values of $q$ required are bigger than the ones for European options due
to the more complicated payoff function. \ The values of $q$ for different
values of $\Delta S_{t}$ is also given in Table \ref{SectionResult}.1.
\ Similarly, the hedging performance of {up-and-in} options, down-and-out and
down-and-in options are given in Figure \ref{SectionResult}.2 and Table
\ref{SectionResult}.1.%

\begin{tabular}
[c]{|llllllll||}\hline
\multicolumn{8}{|l||}{In static hedging of European options}\\
\multicolumn{8}{|l||}{in Figure \ref{SectionResult}.1,}\\\hline
\multicolumn{1}{|l|}{$\Delta S_{t}$} & 10 & 20 & 30 & 40 & 50 & 60 & 70\\
\multicolumn{1}{|l|}{$q$} & 8 & 14 & 20 & 26 & 32 & 36 & 38\\\hline
\multicolumn{8}{|l||}{In dynamic hedging of European options}\\
\multicolumn{8}{|l||}{in Figure \ref{SectionResult}.1,}\\\hline
\multicolumn{1}{|l|}{$\Delta S_{t}$} & 10 & 20 & 30 & 40 & 50 & 60 & 70\\
\multicolumn{1}{|l|}{$q$} & 8 & 14 & 20 & 26 & 32 & 36 & 38\\\hline
\multicolumn{8}{|l||}{In dynamic hedging of up-and-out options}\\
\multicolumn{8}{|l||}{in Figure \ref{SectionResult}.1,}\\\hline
\multicolumn{1}{|l|}{$\Delta S_{t}$} & 10 & 20 & 30 & 40 & 50 & 60 & 70\\
\multicolumn{1}{|l|}{$q$} & 9 & 15 & 22 & 27 & 32 & 36 & 39\\\hline
\end{tabular}%
\begin{tabular}
[c]{|l|lllllll|}\hline
\multicolumn{8}{|l|}{In dynamic hedging of up-and-in options}\\
\multicolumn{8}{|l|}{in Figure \ref{SectionResult}.2,}\\\hline
$\Delta S_{t}$ & 10 & 20 & 30 & 40 & 50 & 60 & 70\\
$q$ & 9 & 16 & 22 & 28 & 32 & 36 & 39\\\hline
\multicolumn{8}{|l|}{In dynamic hedging of down-and-out options}\\
\multicolumn{8}{|l|}{in Figure \ref{SectionResult}.2,}\\\hline
$\Delta S_{t}$ & 10 & 20 & 30 & 40 & 50 & 60 & 70\\
$q$ & 8 & 14 & 20 & 26 & 32 & 36 & 38\\\hline
\multicolumn{8}{|l|}{In dynamic hedging of down-and-in options}\\
\multicolumn{8}{|l|}{in Figure \ref{SectionResult}.2,}\\\hline
$\Delta S_{t}$ & 10 & 20 & 30 & 40 & 50 & 60 & 70\\
$q$ & 9 & 16 & 22 & 28 & 32 & 36 & 39\\\hline
\end{tabular}

Table \ref{SectionResult}.1: \ The values of $q$ for given $\Delta S_{t}$ in
static hedging of European options, dynamic hedging of European, up-and-out,
up-and-in, down-and-out and down-and-in options.

\bigskip

The performance of hedging some other exotic options, such as lookback options
and Asian options, can be obtained similarly since we employ Monte Carlo
simulation in calculating the option prices. \ Recall in Section
\ref{SectionImple}, as $N$ increases, the number of derivatives that can be
calculated increases. \ The results show that $q$ increases rapidly with
increasing $\Delta S_{t}$. \ Note that the bigger the value of $\Delta S_{t}$,
the slower the convergence rate of Taylor expansion and this is why dynamic
hedging is more popular in the literature. \ From our simulation results, we
note that $\frac{D_{2}^{i}F\left(  t+\Delta t,S_{t}\right)  }{i}$ become very
small as $i$ increases, but the value of $\left(  \Delta S_{t}\right)  ^{i}$
increases very rapidly. \ Therefore, we cannot ignore the terms $\frac
{D_{2}^{i}F\left(  t+\Delta t,S_{t}\right)  }{i!}\left(  \Delta S_{t}\right)
^{i}$. \ To enable perfect hedging using moment swaps, power jump assets or
some other traded derivatives depending on the same underlying asset, the
market has to allow trading in these financial derivatives in a unit as small
as $\frac{D_{2}^{i}F\left(  t+\Delta t,S_{t}\right)  }{i}$.

In summary, as long as we can find the $q$ such that the Taylor approximations
are accurate for all possible values of $\Delta S_{t}$ under consideration,
the perfect hedging using moment swaps, power jump assets or other traded
derivatives depending on the same underlying asset works very well.

To show the trading strategy is applicable to real life data, we fit the VG
model to European option price on FTSE index and derive a static hedging
strategy on a one year European option. \ On 4th January 2007 and 4th January
2008, the spot FTSE 100 index are 6287 and 6348.5, respectively. \ The change
in value of the underlying, $\Delta S$, is therefore 61.5. \ We apply our
hedging strategy to the one year European option on 4th January 2007 and show
how hedging can be achieved. \ The option with strike 6287 is worth 410.3 on
4th January 2007, where the risk-free interest rate is 5.43\%, the dividend on
the FTSE 100 index is 3.51\% and the implied volatility is 14.65\%. \ We fit
these data using the VG model and obtain the parameter values: $\theta=$
$-0.2721$, $\nu=$ 0.3032 and $\sigma=$ 3.02\%. \ The Monte Carlo (MC)
simulated option price using these parameters is 410.914. \ The pricing error
due to calibration and simulation is then $0.614$. \ At maturity, the option
is in the money and the payoff is $(6348.5-6287)=61.5$. \ Therefore, the
change of value of the option is $61.5-410.914=-349.414$ according to the MC
calculation. \ The hedging performance is given in Table \ref{SectionResult}%
.2. \ The first column shows the number of terms used in the Taylor expansion,
the second column shows the value of the derivative, $D_{2}^{\left(  i\right)
}\left(  t+\Delta t,\Delta S\right)  $ and the third column shows the
approximated price. \ We see that 12 terms are needed to obtain the change in
option price, $-349.414$. \ Perfect hedging is achieved according to the MC
price. \ It shows that the hedging strategy works very well in replicating the
MC price and the hedging error in this example is entirely due to calibration
and simulation of the VG model, which is out of the scope of this paper. \ We
note that the contributions of the odd number terms except the first term are
almost negligible and can be ignored. \ Therefore we can reduce the number of
instruments invested in this case.

$%
\begin{tabular}
[c]{|lll||}\hline
\multicolumn{1}{|l|}{1} & \multicolumn{1}{l|}{0.5} & -380.164\\
\multicolumn{1}{|l|}{2} & \multicolumn{1}{l|}{0.01107} & -338.294\\
\multicolumn{1}{|l|}{3} & \multicolumn{1}{l|}{-8.97421e-015} & -338.294\\
\multicolumn{1}{|l|}{4} & \multicolumn{1}{l|}{-9.39954e-007} & -351.741\\
\multicolumn{1}{|l|}{5} & \multicolumn{1}{l|}{9.30204e-019} & -351.741\\\hline
\end{tabular}%
\begin{tabular}
[c]{|lll||}\hline
6 & \multicolumn{1}{l|}{4.80557e-011} & -349.141\\
7 & \multicolumn{1}{l|}{-4.05377e-023} & -349.141\\
8 & \multicolumn{1}{l|}{-1.4317e-015} & -349.434\\
9 & \multicolumn{1}{l|}{9.58199e-028} & -349.434\\
10 & \multicolumn{1}{l|}{2.6928e-020} & -349.413\\\hline
\end{tabular}%
\begin{tabular}
[c]{|lll|}\hline
11 & \multicolumn{1}{l|}{-1.39873e-032} & -349.413\\
12 & \multicolumn{1}{l|}{-3.40129e-025} & -349.414\\
13 & \multicolumn{1}{l|}{1.36317e-037} & -349.414\\
14 & \multicolumn{1}{l|}{3.01623e-030} & -349.414\\
15 & \multicolumn{1}{l|}{-9.35744e-043} & -349.414\\\hline
\end{tabular}
$

Table \ref{SectionResult}.2: The performance of the hedging strategy on a one
year European option price on FTSE 100 index on 4th January 2007.

\section{Conclusion\label{SectionConclusion}}

In this paper, we provided some perfect hedging strategies and minimal
variance portfolios in a L\'{e}vy market. \ Many financial institutions hold
derivative securities in their portfolios, and frequently these securities
need to be hedged for extended periods of time. Failure to hedge properly can
expose an institution to sudden swings in the values of derivatives, such as
options, resulting from large, unanticipated changes in the levels or
volatilities of the underlying asset. \ Research in the techniques employed
for hedging derivative securities is therefore of crucial importance. \ Under
the assumption of the famous Black-Scholes model, the market is complete and
an European option can be hedged perfectly by investing in a risk-free bank
account and the underlying stock.\ \ However, there is statistical evidence,
such as the volatility smile, that the Black-Scholes model is not sufficiently
flexible to model the price process. \ As a result, the study of L\'{e}vy
process, which is a generalisation of Brownian motion with jumps, has become
increasingly important in mathematical finance. \ If the underlying asset is
driven by a L\'{e}vy process, the market is not complete, that is, a
contingent claim cannot be hedged using only a risk-free bank account and the
underlying asset. \ By applying a Taylor expansion to the pricing formulae, we
derived dynamic perfect hedging strategies of European and some exotic options
by trading in moment swaps, power jump assets or certain traded derivatives
depending on the same underlying asset. \ In the case of European options,
static hedging can also be achieved. \ We extended the delta and gamma hedging
strategies to higher moment hedging by investing in other traded derivatives
depending on the same underlying asset. \ We demonstrated how to use the
minimal variance portfolios derived by \cite{bdlop03} to hedge the higher
order terms in the Taylor expansion, investing only in a risk-free bank
account, the underlying asset and, potentially, variance swaps. \ We
explicitly addressed numerical issues in the procedures, such as the
approximation of the derivatives in the Taylor expansion, as well as
investigated the performance of the hedging strategies. \ If as many
derivatives as the Taylor expansion needed for accuracy can be determined and
the financial derivatives required to hedge are available in the specified
amounts, perfect hedging is possible.

\setcounter{section}{0} \renewcommand{\thesection}{\Alph{section}}
\renewcommand{\thesubsection}{\Alph{section}.\arabic{subsection}} \numberwithin{equation}{section}%

\small

\begin{center}
{\large {\textbf{APPENDICES}}}
\end{center}

\setcounter{equation}{0}\renewcommand{\theequation}{A.\arabic{equation}}

\section{Proof of Propositions and Lemma}

\subsection{Proof of Proposition \ref{Proposition3}\label{AppendixProp3}}

The initial investment at time $t$ is%
\[
C_{i}\left\{  S_{t}^{i}e^{-r\left(  t+\Delta t\right)  }T_{t}^{\left(
i\right)  }+\frac{S_{t}^{i}e^{-r\left(  t+\Delta t\right)  }T_{t}^{\left(
i\right)  }}{e^{r\Delta t}-1}+\frac{S_{t}^{i}\left[  -e^{-rt}T_{t}^{\left(
i\right)  }+m_{i}\Delta t\right]  }{e^{r\Delta t}-1}\right\}  .
\]
At maturity, the value of the portfolio is equal to%
\[
C_{i}S_{t}^{i}\left\{  e^{-r\left(  t+\Delta t\right)  }T_{t+\Delta
t}^{\left(  i\right)  }+\frac{e^{r\Delta t}}{e^{r\Delta t}-1}\left\{
e^{-r\left(  t+\Delta t\right)  }T_{t}^{\left(  i\right)  }-e^{-rt}%
T_{t}^{\left(  i\right)  }+m_{i}\Delta t\right\}  \right\}  .
\]
Hence, by equation (\ref{PJAD}), the change of value of the portfolio equals
\[
C_{i}\left\{  S_{t}^{i}e^{-r\left(  t+\Delta t\right)  }T_{t+\Delta
t}^{\left(  i\right)  }+S_{t}^{i}\left[  -e^{-rt}T_{t}^{\left(  i\right)
}+m_{i}\Delta t\right]  \right\}  .
\]

\subsection{Proof of Proposition \ref{PropositionMVP}%
\label{AppendixProofProp1}}

Let
\begin{equation}
\xi=\xi^{0}+\sum_{j=1}^{k}\int_{0}^{T}\varphi_{j}\left(  s\right)
\mathrm{d}S_{j}\left(  s\right)  . \label{errorTerm}%
\end{equation}
where $\xi^{0}$ denotes the difference of value between $\xi$ and $\sum
_{j=1}^{k}\int_{0}^{T}\varphi_{j}\left(  s\right)  \mathrm{d}S_{j}\left(
s\right)  $ for the portfolio $\varphi=\left(  \varphi_{1},...,\varphi
_{k}\right)  $. \ By the results of \cite[Section 4.2]{ms95}, the Hilbert
space argument in \cite[Theorem 2.3]{bdlop03} and equation (\ref{decomp}), the
following orthogonality condition is satisfied: $E\left[  \left(  \xi-\hat
{\xi}\right)  \Theta\right]  =E\left[  \left\{  \xi^{0}-E\left[  \xi\right]
\right\}  \Theta\right]  $ $=E\left[  \xi^{0}\Theta\right]  -E\left[
\xi\right]  E\left[  \Theta\right]  =0,$ where
\begin{equation}
\Theta=\sum_{j=1}^{k}\int_{0}^{T}\theta_{j}\left(  s\right)  \sigma_{j}%
S_{j}\left(  s_{-}\right)  \mathrm{d}W_{j}\left(  s\right)  +\sum_{j=1}%
^{k}\int_{0}^{T}\int_{\mathbb{R}}x\theta_{j}\left(  s\right)  S_{j}\left(
s_{-}\right)  \tilde{N}_{j}\left(  \mathrm{d}s,\mathrm{d}x\right)
\label{Theta}%
\end{equation}
for all $\theta=\left(  \theta_{1},...,\theta_{k}\right)  \in\mathcal{A}$.
\ Since $E\left[  \Theta\right]  =0,$ we have $E\left[  \xi^{0}\Theta\right]
=0.$ \ From (\ref{decomp}) and (\ref{S2}),%
\[
\sum_{j=1}^{k}\int_{0}^{T}\varphi_{j}\left(  s\right)  \mathrm{d}S_{j}\left(
s\right)  =\sum_{j=1}^{k}\int_{0}^{T}\varphi_{j}\left(  s\right)  S_{j}\left(
s_{-}\right)  b_{j}\mathrm{d}s+\sum_{j=1}^{k}\int_{0}^{T}\varphi_{j}\left(
s\right)  \sigma_{j}S_{j}\left(  s_{-}\right)  \mathrm{d}W_{j}\left(
s\right)
\]%
\[
+\sum_{j=1}^{k}\int_{0}^{T}\int_{\mathbb{R}}x\varphi_{j}\left(  s\right)
S_{j}\left(  s_{-}\right)  \tilde{N}_{j}\left(  \mathrm{d}s,\mathrm{d}%
x\right)  .
\]
Hence, from (\ref{f1f2}) and (\ref{errorTerm}),
\begin{align*}
\xi^{0}  &  =E\left[  \xi\right]  -\sum_{j=1}^{k}\int_{0}^{T}\varphi
_{j}\left(  s\right)  S_{j}\left(  s_{-}\right)  b_{j}\mathrm{d}s+\sum
_{j=1}^{k}\int_{0}^{T}\left(  \frac{1}{\sigma_{j}}f_{1}\left(  \xi;s,j\right)
-\varphi_{j}\left(  s\right)  S_{j}\left(  s_{-}\right)  \right)  \sigma
_{j}\mathrm{d}W_{j}\left(  s\right) \\
&  +\sum_{j=1}^{k}\int_{0}^{T}\int_{\mathbb{R}}\left(  f_{2}\left(
\xi;s,x,j\right)  -x\varphi_{j}\left(  s\right)  S_{j}\left(  s_{-}\right)
\right)  \tilde{N}_{j}\left(  \mathrm{d}s,\mathrm{d}x\right)  .
\end{align*}
Hence, from (\ref{Theta}) and the well-known isometry, see \cite{iw89}, we
have%
\begin{align*}
E\left[  \xi^{0}\Theta\right]   &  =\sum_{j=1}^{k}E\left[  \int_{0}^{T}%
\theta_{j}\left(  s\right)  S_{j}\left(  s_{-}\right)  \left\{  \left(
f_{1}\left(  \xi;s,j\right)  -\sigma_{j}\varphi_{j}\left(  s\right)
S_{j}\left(  s_{-}\right)  \right)  \sigma_{j}\right.  \right. \\
&  +\left.  \left.  \int_{\mathbb{R}}x\left(  f_{2}\left(  \xi;s,x,j\right)
-x\varphi_{j}\left(  s\right)  S_{j}\left(  s_{-}\right)  \right)  \nu
_{j}\left(  \mathrm{d}x\right)  \right\}  \mathrm{d}s\right]  =0.
\end{align*}%
\begin{align*}
&  \Rightarrow f_{1}\left(  \xi;s,j\right)  \sigma_{j}+\int_{\mathbb{R}}%
xf_{2}\left(  \xi;s,x,j\right)  \nu_{j}\left(  \mathrm{d}x\right)
=\varphi_{j}\left(  s\right)  S_{j}\left(  s\right)  \left\{  \sigma_{j}%
^{2}+\int_{\mathbb{R}}x^{2}\nu_{j}\left(  \mathrm{d}x\right)  \right\} \\
\varphi_{j}\left(  s\right)   &  =\left[  f_{1}\left(  \xi;s,j\right)
\sigma_{j}+\int_{\mathbb{R}}xf_{2}\left(  \xi;s,x,j\right)  \nu_{j}\left(
\mathrm{d}x\right)  \right]  /\left[  \left\{  \sigma_{j}^{2}+\int
_{\mathbb{R}}x^{2}\nu_{j}\left(  \mathrm{d}x\right)  \right\}  S_{j}\left(
s\right)  \right]  .
\end{align*}

\subsection{Proof of Proposition \ref{PropositionMVPwithout}%
\label{AppendixMVPwithout}}

From equation (\ref{PJADi}), the term $\sum_{i=2}^{q}C_{i}S_{t}^{i}m_{i}\Delta
t$ can be hedged by investing
\[
\sum_{i=2}^{q}\frac{C_{i}S_{t}^{i}m_{i}\Delta t}{\exp\left(  r\Delta t\right)
-1}%
\]
in a risk-free bank account. \ To hedge the term $\sum_{i=2}^{q}C_{i}S_{t}%
^{i}\int_{t}^{t+\Delta t}\mathrm{d}Y_{s}^{\left(  i\right)  },$ we let
\[
\xi=\sum_{i=2}^{q}\int_{t}^{t+\Delta t}C_{i}S_{t}^{i}\mathrm{d}Y_{s}^{\left(
i\right)  }=\sum_{i=2}^{q}\int_{t}^{t+\Delta t}\int_{\mathbb{R}}C_{i}S_{t}%
^{i}x^{i}\tilde{N}\left(  \mathrm{d}s,\mathrm{d}x\right)
\]
by (\ref{YN}) and let the minimal variance portfolio to hedge $\xi$ be
$\hat{\xi}=E\left[  \xi\right]  +\int_{t}^{t+\Delta t}\varphi_{s}%
\mathrm{d}S_{s}=\int_{t}^{t+\Delta t}\varphi_{s}\mathrm{d}S_{s}$ since
$E\left[  \xi\right]  =0$. \ Hence, using Proposition \ref{PropositionMVP} and
equation (\ref{f1f2}) by putting $f_{1}\left(  \xi;s,j\right)  =0$ and
$f_{2}\left(  \xi;s,x,j\right)  =\sum_{i=2}^{q}C_{i}S_{t}^{i}x^{i}$, we have
\[
\varphi_{s}=\frac{\int_{\mathbb{R}}\sum_{i=2}^{q}C_{i}S_{t}^{i}x^{i+1}%
\nu\left(  \mathrm{d}x\right)  }{\left[  \sigma^{2}+\int_{\mathbb{R}}x^{2}%
\nu\left(  \mathrm{d}x\right)  \right]  S_{s}}.
\]
Hence, to hedge the terms $\sum_{i=2}^{q}Q_{i}\ $by minimal variance
portfolio, we need to invest the amount
\[
\sum_{i=2}^{q}\frac{C_{i}S_{t}^{i}m_{i}\Delta t}{\exp\left(  r\Delta t\right)
-1}%
\]
in a risk-free bank account and buy
\[
\frac{\int_{\mathbb{R}}\sum_{i=2}^{q}C_{i}S_{t}^{i}x^{i+1}\nu\left(
\mathrm{d}x\right)  }{\left[  \sigma^{2}+\int_{\mathbb{R}}x^{2}\nu\left(
\mathrm{d}x\right)  \right]  S_{t}}=\frac{\sum_{i=2}^{q}C_{i}S_{t}%
^{i-1}m_{i+1}}{\left[  \sigma^{2}+m_{2}\right]  }%
\]
amount of the underlying stock, $S_{t}$, where $m_{i}$ are defined in
(\ref{mean}). \ 

\renewcommand{\theequation}{B.\arabic{equation}}

\section{Central difference approximation of arbitrary
degree\label{centralDiff}}

\cite[Section 1]{ko03} showed that Taylor's series based central difference
approximation of arbitrary $p$-th degree derivative of a function $f\left(
t\right)  $ at $t=t_{0}$ can be written for an order $2N$ as%
\begin{equation}
f_{0}^{\left(  p\right)  }=\frac{1}{T^{p}}\sum_{k=-N}^{N}d_{k}^{\left(
p\right)  }f_{k}, \label{fp}%
\end{equation}
where $T$ is the sampling period, $2N+1$ is the number of nodes used in the
approximation, $f_{k}$ denotes the value of function $f\left(  t\right)  $ at
$t=t_{0}+kT$, $2N$ is an integer bigger than $p$ and $d_{0}^{\left(  p\right)
}=0$ if $p$ is odd, otherwise $d_{0}^{\left(  p\right)  }=-2\sum_{k=1}%
^{N}d_{k}^{\left(  p\right)  },$ and%
\begin{equation}
d_{k}^{\left(  p\right)  }=\left(  -1\right)  ^{k+c_{1}}\frac{p!}{k^{1+c_{2}}%
}C_{N,k}\sum_{i}\frac{1}{X\left(  i\right)  ^{2}},\ \ \ \text{for
}k=-N,-N+1,...,-1,1,...,N-1,N, \label{dp2}%
\end{equation}
$C_{N,k}=\frac{N!^{2}}{\left(  N-k\right)  !\left(  N+k\right)  !}$, $c=$
largest integer less than or equal to $\left(  p-1\right)  /2$, $c_{1}=1$ if
$c $ is even, otherwise $c_{1}=0$, $c_{2}=1$ if $p$ is even, otherwise
$c_{2}=0$, and the vector $X$ is generated in the following way:

1. Take a vector $Y$ containing all integers from $1$ to $N$ except
$\left\vert k\right\vert $ (in \cite[p. 121]{ko03}, it was except $k$, but
from the derivation of the formula, it should be $\left\vert k\right\vert $).

2. The vector $X$ contains the product of all the possible combinations of
length $c$ in $Y$. \ 

\bibliographystyle{authordate4}
\bibliography{MyReferences,prm,sk}

\end{document}